\documentclass[11pt,a4paper]{article}

\usepackage[utf8]{inputenc}
\usepackage[T1]{fontenc}

\usepackage{fullpage}

\usepackage[small]{caption}
\usepackage{subcaption}

\usepackage{amsmath}
\usepackage{amsfonts}
\usepackage{amssymb}
\usepackage{amsthm}
\usepackage{amstext}
\usepackage{thmtools}

\usepackage{csquotes}
\usepackage{xcolor}
\usepackage{graphicx}

\usepackage{mathtools}
\usepackage{xspace}
\usepackage{enumerate}
\usepackage[ruled]{algorithm2e}
\usepackage{paralist}
\usepackage{bm}

\usepackage{algorithm}
\usepackage{algorithmic}
\usepackage{url}
\usepackage{hyperref}
\usepackage[capitalize,noabbrev]{cleveref}

\usepackage{libertine}
\usepackage[libertine]{newtxmath}
\usepackage{microtype}

\theoremstyle{plain}
\newtheorem{theorem}{Theorem}[section]
\newtheorem{proposition}[theorem]{Proposition}
\newtheorem{lemma}[theorem]{Lemma}
\newtheorem{corollary}[theorem]{Corollary}
\theoremstyle{definition}

\theoremstyle{remark}

\usepackage{bm}
\usepackage{dsfont}
\usepackage{xspace}
\usepackage{thm-restate}
\usepackage{enumerate}
\usepackage{paralist}
\usepackage{soul}

\usepackage{acro}
\DeclareAcronym{ML}{short=ML,long=machine learning,short-indefinite=an}
\DeclareAcronym{SoC}{short=SoC,long=system-on-chip,long-plural-form=systems-on-chip,short-indefinite=an}

\theoremstyle{plain}

\newtheorem{observation}[theorem]{Observation}

\newcommand{\littlec}{\mbox{\emph{LITTLE}}\xspace}
\newcommand{\bigc}{\emph{big}\xspace}
\newcommand{\hikey}{\emph{HiKey~970}\xspace}
\newcommand{\splash}{\mbox{\emph{SPLASH-3}}\xspace}
\newcommand{\parsec}{\mbox{\emph{PARSEC-3.0}}\xspace}
\newcommand{\polybench}{\emph{Polybench}\xspace}

\DeclarePairedDelimiter\floor{\lfloor}{\rfloor}
\DeclarePairedDelimiter{\abs}{\lvert}{\rvert}
\DeclareMathOperator*{\argmax}{arg\,max}
\DeclareMathOperator*{\argmin}{arg\,min}
\newcommand{\bigO}{\mathcal{O}}

\newcommand{\ind}{\ensuremath{\mathds{1}}}

\newcommand{\alg}{{\normalfont\textsc{Alg}}}
\newcommand{\opt}{{\normalfont\textsc{Opt}}}

\newcommand{\hy}{\bm\hat{y}}

\newcommand{\hr}{\bm\hat{r}}
\newcommand{\hs}{\bm\hat{s}}

\newcommand{\hQ}{\bm\hat{Q}}

\newcommand{\dense}{\delta}
\newcommand{\hdense}{\bm\hat{\delta}}

\newcommand{\dualSa}{\bm\bar{a}}
\newcommand{\dualSb}{\bm\bar{b}}
\newcommand{\dualSc}{\bm\bar{c}}
\newcommand{\dualVa}{a}
\newcommand{\dualVb}{b}
\newcommand{\dualVc}{c}

\title{Speed-Oblivious Online Scheduling: \\ Knowing (Precise) Speeds is not Necessary}

\author{Alexander Lindermayr~\thanks{University of Bremen, Faculty of Mathematics and Computer Science, Germany. \{linderal,nmegow\}@uni-bremen.de}\and Nicole Megow~\footnotemark[1] \and Martin Rapp~\stepcounter{footnote}\stepcounter{footnote}\stepcounter{footnote}\stepcounter{footnote}\stepcounter{footnote}\thanks{Faculty for Informatics, Karlsruhe Institute of Technology, Germany. martin.rapp@kit.edu}}

\begin{document}

\maketitle

\thispagestyle{empty}

\begin{abstract}
We consider online scheduling on unrelated (heterogeneous) machines in a \emph{speed-oblivious} setting, where an algorithm is unaware of the exact job-dependent processing speeds. 
We show strong impossibility results for clairvoyant and non-clairvoyant algorithms and overcome them in models inspired by practical settings: 
(i)~we provide competitive \emph{learning-augmented} algorithms, assuming that (possibly erroneous) predictions on the speeds are given, and 
(ii)~we provide competitive algorithms for the \emph{speed-ordered} model, where a single global order of machines according to their unknown job-dependent speeds is known. 
We prove strong theoretical guarantees and evaluate our findings on a representative heterogeneous multi-core processor. 
These seem to be the first empirical results for scheduling algorithms with predictions that are evaluated 
in a non-synthetic hardware environment.
\end{abstract}

\section{Introduction}
Heterogeneous processors are getting more and more common in various domains. 
For several years now, efficiency and performance gains in smartphone chips have depended crucially on the combination of
high-performance and low-performance (but energy-efficient) cores \cite{arm_big_little_whitepaper}.
Heterogeneity has recently been introduced also to the desktop market with Intel Alder Lake (Q1'2022)~\cite{alderlake} and AMD Zen 5 (announced for 2023).
Further,
jobs differ in their instruction mix and memory access patterns, and hence may not benefit uniformly from the high-performance cores, which typically feature larger caches, out-of-order execution, and a higher CPU frequency.
\Cref{fig:characterization} shows job-dependent speed varieties in common benchmark suites (\parsec, \splash, \polybench) running on \bigc and \littlec cores of a Kirin~970 smartphone \ac{SoC} with Arm big.LITTLE architecture. 

\begin{figure}
    \centering
    \includegraphics{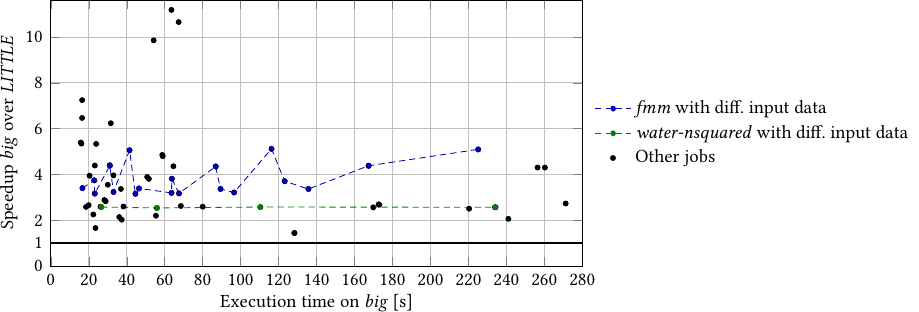}
    \caption{
        The execution time and speedup of the \bigc over \littlec cores on an Arm big.LITTLE heterogeneous processor varies strongly between jobs and different input data.
        Variations of the speedup w.r.t.~input data are large for some jobs (e.g., \emph{water-nsquared}) but small for others (e.g., \emph{fmm}).
        }
    \label{fig:characterization}
\end{figure}

These advances show the demand for schedulers that respect {\em job-dependent heterogeneity}. Formally, the \emph{(processing) speed} $s_{ij}$ of job $j$ on machine~$i$ is the amount of processing that $j$ receives when running on~$i$ for one time unit.     
Despite the relevance of values $s_{ij}$ for high-performance scheduling, there is a big discrepancy between how theory and practice handle them: while scheduling theory most commonly assumes that speeds are known to an algorithm, this is typically not the case in practice. Hence, algorithms that perform well in theory are often not applicable in practice.

In this
work, we propose new models and algorithms to bridge this gap. In particular, we introduce \emph{speed-oblivious} algorithms, which do not rely on knowing (precise) speeds. Thereby we focus on (non-)clairvoyant scheduling subject to minimizing the total (weighted) completion time.

Formally, an instance of this scheduling problem is composed of a set~$J$ of~$n$ jobs, a set~$I$ of ~$m$ machines, and a time-discretization. The characteristics of a job~$j \in J$ are its processing requirement~$p_j$, its weight~$w_j$, and for every machine~$i \in I$ its individual processing speed~$s_{ij} > 0$. A job $j$ arrives online at its release date $r_j$, i.e., an algorithm is unaware of its existence before that time.
A schedule assigns for every unfinished job~$j \in J$ and for every machine~$i \in I$ at any time~$t \geq r_j$ a \emph{(machine) rate}~$y_{ijt} \in [0,1]$, which induces the progress~$q_{jt} = \sum_i s_{ij} y_{ijt}$ of~$j$ at time~$t$. 
The completion time~$C_j$ of a job~$j$ in a fixed schedule is the first point in time~$t$ that satisfies~$\sum_{t'=r_j}^{t} q_{jt'} \geq p_j$.
A schedule is feasible if there exists a progress-preserving actual schedule, where at any infinitesimal time a job is being processed on at most one machine. This applies if the rates satisfy~$\sum_{i \in I} y_{ijt} \leq 1$ for all~$j \in J$ and~$\sum_{j \in J} y_{ijt} \leq 1$ for all~$i \in I$ at any time~$t$~\cite{ImKM18}. 
The task is to compute a feasible schedule that minimizes~$\sum_{j \in J} w_j C_j$.

An algorithm is called {\em non-migratory}, if it assigns for each job~$j$ positive rates only on a single machine~$i_j$, and {\em migratory} otherwise. Further, it is called \emph{non-preemptive} if for all jobs~$j$, machines~$i$, and times~$t$, a rate~$y_{ijt} > 0$ implies~$y_{ijt'} = 1$ for all times~$t'$ with~$t \leq t' \leq C_j$. We say that the machines are \emph{related} if $s_{i} = s_{ij}$ for all jobs $j$ and machines $i$, i.e., speeds are not job-dependent.

\paragraph{Models and state-of-the-art in theory}
Scheduling jobs (offline) on machines with job-dependent heterogeneity (called \emph{unrelated} machine scheduling) to minimize the total weighted completion time is a prominent NP-hard problem; several approximation algorithms are known, e.g.,~\cite{HallSSW97,SchulzS02,Li20,BansalSS21,ImL23}.  Well-studied {\em online} models include \emph{online} job  arrival~\cite{PruhsST04}, i.e., a job is unknown to an algorithm until its  release date~$r_j$, and \emph{non-clairvoyance}~\cite{MotwaniPT94}, i.e., an algorithm has no knowledge about the total processing requirement~$p_j$ of a job (as opposed to \emph{clairvoyant} schedulers).
In particular, online algorithms cannot revert previous decisions.
The performance of an online algorithm is typically evaluated by its \emph{competitive ratio} \cite{Borodin98}, i.e., the worst-case ratio between the algorithm's objective value and the optimal objective value (given full information upfront) for every instance. We say that an algorithm is~$\rho$-competitive if its competitive ratio is at most~$\rho$.
Known online results include~\cite{HallSSW97,ChadhaGKM09,AnandGK12,ImKMP14,ImKM18,GuptaMUX20,DBLP:phd/dnb/Jager21,BienkowskiKL21,LindermayrM22}. 

To the best of our knowledge, unrelated machine scheduling has been studied only in a \emph{speed-aware} setting, where an algorithm knows the speeds~$s_{ij}$ for available jobs. It is not difficult to see that there are prohibitive lower bounds for speed-oblivious scheduling on (un-)related machines: consider an instance with a single unit-sized job~$j$ which makes substantial progress only on one machine. This means that in the worst-case, the first $m-1$ machines tried by the algorithm have speed $\epsilon$ and $j$ makes no substantial progress. Thus, the algorithm spends at least $m$ time units to complete it. Knowing this fast machine upfront allows an optimal solution to complete the job immediately.
This implies a competitive ratio of at least~$\Omega(m)$ for~$m$ machines:

\begin{observation}
Any speed-oblivious algorithm has a competitive ratio of at least $\Omega(m)$ for minimizing the total (weighted) completion time on $m$ related machines, even if the algorithm is clairvoyant.
\end{observation}

\paragraph{Models and state-of-the-art in practice}
Practical scheduling algorithms commonly operate in open systems \cite{open_system}, where jobs arrive online, are non-clairvoyant, and, in contrast to the assumption in theory, their exact processing speeds on every core are \emph{unknown} upfront.
Therefore, state-of-the-practice schedulers usually ignore heterogeneity between jobs, e.g., Linux Energy-Aware Scheduling \cite{linux_eas}.
State-of-the-art schedulers rely on prior knowledge about jobs \cite{khdr2015thermal}, which is not always available, or rely on predictions of job characteristics to leverage this information gap.
Such predictions could be based on prior executions of repeating jobs or on machine-learning-based techniques \cite{gupta2018staff,snuca_tc_prediction}. 
They are often quite precise, but can be highly inaccurate 
due to varying and unpredictable input data as shown in \Cref{fig:characterization}.
To the best of our knowledge, all these approaches are evaluated only empirically. In particular, there are no theoretical guarantees on the performance in worst-case scenarios or with respect to a prediction's quality.

\subsection{Our Results}
We initiate the theoretical study of speed-oblivious algorithms. Since strong lower bounds rule out good worst-case guarantees for speed-oblivious unrelated machine scheduling without further assumptions, we propose two (new) models which are motivated by data-driven machine-learned models and modern heterogeneous hardware architectures: 

\begin{itemize}
  \item \textbf{Speed predictions}
  give algorithms access to values~$\hs_{ij}$ for every machine~$i$ at the release date of every job~$j$.
  We measure the accuracy of such a prediction by the \emph{distortion error}~$\mu$, where~$\mu = \mu_1 \cdot \mu_2$ and
  \[
    \mu_1 = \max_{i \in I,j \in J} \left\{ \frac{\hs_{ij}}{s_{ij}} \right\} \text{ and } \mu_2 = \max_{i \in I,j \in J} \left\{ \frac{s_{ij}}{\hs_{ij}}\right\}.
    \]

  \item \textbf{Speed-ordered machines}
assume an order on~$I$ such that for all~$i,i' \in I$ and jobs~$j \in J$  holds~$s_{ij} \geq s_{i'j}$ if and only if~$i \leq i'$. Algorithms are aware of this order.
\end{itemize}

Finally, we compare algorithms for these models with heuristics in experiments on an actual modern heterogeneous chip.
These are the first empirical results which show the benefit of learning-augmented algorithms and validate theoretical findings on \emph{real} hardware. In particular, we initiate the investigation in practical applicability of theoretical scheduling algorithms for actual realistic hardware environments.

We now give a more detailed overview of our results.

\paragraph{Learning-augmented algorithms for speed predictions}

We provide the first learning-augmented algorithms with job-dependent speed predictions and prove error-dependent performance guarantees w.r.t.\ the distortion error $\mu$. This gives formal evidence on why algorithms perform well in practice, even if the assumed speeds slightly diverge from the true speeds. We further show that a competitive ratio linear in $\mu$ is best possible, even for migratory algorithms and related machines.
We emphasize that the algorithms do not have access to~$\mu$ upfront for the given instance.

\begin{theorem}
  For minimizing the total weighted completion time on unrelated machines, there exist speed-oblivious online algorithms with speed predictions that are
\begin{compactenum}[(i)]
    \item clairvoyant and~$8\mu$-competitive,
    \item clairvoyant, non-preemptive and~$7.216\mu^2$-competitive, 
    \item non-clairvoyant and~$108\mu$-competitive.
  \end{compactenum}
\end{theorem}

For~$(i)$, we design a novel and efficient clairvoyant algorithm, which might be of independent interest. It always schedules the subset of jobs that maximizes the total (predicted) density in a feasible job-to-machine assignment, where the density of a job~$j$ on machine~$i$ is equal to~$\frac{w_j s_{ij}}{p_j}$. We show that it is $8$-competitive in the speed-aware setting. Interestingly, this algorithm reduces to Smith's rule on a single machine~\cite{smith1956various} and preemptive variants \cite{SchulzSk2002,MegowS04}.

On the technical side, we prove upper bounds on the competitive ratios using the \emph{dual-fitting} technique~\cite{JainMMSV03,AnandGK12}. There, we lower bound the optimal solution by the dual of a linear programming (LP) relaxation, and then show that a specific feasible dual assignment has an objective value which is close to the algorithm's objective value. The main difficulty is therefore to come up with good dual assignment. For~$(i)$, we present a new dual setup, which we believe could be helpful for future dual-fitting approaches. The algorithms and proofs for~$(ii)$ and~$(iii)$ are are inspired by previous work (Greedy WSPT~\cite{GuptaMUX20}, Proportional Fairness~\cite{ImKM18}). However, for~$(iii)$ we achieve better constants via optimized duals, even for the speed-aware case.
In all proofs, we use scalable properties of duals to convert bad decisions due to imprecise predictions into scaled bounds on the competitive ratio.

\paragraph{Novel algorithms for speed-ordered machines}
The strong lower bound of $\Omega(m)$ on the competitive ratio for speed-oblivious algorithms for $m$ machines crucially relies on accelerating the machine that an algorithm tries last. This argument becomes infeasible in the speed-ordered setting, because the machines are distinguishable upfront. Designing an algorithm is yet still challenging, as precise factors between speeds remain unknown.
On the negative side, we show that any constant-competitive algorithm must migrate jobs. This is even true for clairvoyant algorithms and related machines. On the positive side, we present two algorithms: 

\begin{theorem}
There is a clairvoyant speed-oblivious online algorithm for minimizing the total weighted completion time on speed-ordered related machines with a competitive ratio of at most $8$.
\end{theorem}

We show that this algorithm is not competitive on unrelated machines. Somewhat surprisingly, our non-clairvoyant algorithm achieves non-trivial competitive ratios for both related and unrelated machines, as the following theorem states.

\begin{theorem}
There is a non-clairvoyant speed-oblivious online algorithm for minimizing the total completion time 
\begin{compactenum}[(i)]
\item on speed-ordered related machines with a competitive ratio of at most $216$, and
\item on speed-ordered unrelated machines with a competitive ratio of $\Theta(\log(\min\{n,m\}))$.
\end{compactenum}
\end{theorem}

A crucial observation for deriving these algorithms is that in the speed-ordered setting certain speed-aware algorithms use strategies which can be formulated \emph{even without} precise speed values. An additional challenge is the few-job regime, i.e., there are less jobs than machines, where we have to ensure that the algorithms prefer the fast machines.

\subsection{Further Related Work}

Uncertainty about machine speeds or, generally, the machine environment, have hardly been studied in scheduling theory. Some works consider scheduling with unknown non-availability periods, i.e., periods with speed~$0$~\cite{AlbersS01,DiedrichJST09}, permanent break-downs of a subset of machines~\cite{SteinZ20}, or more generally arbitrarily changing machine speed for a single machine~\cite{EpsteinLMMMSS12}, but not on heterogenous machines. In scheduling with testing, unknown processing requirements of a job (and thus its machine-dependent speed) can be explored by making queries, e.g.,~\cite{DuerrEMM20,albersE2020,ArantesBKLLS18}, but also here heterogenous processors are not considered.

Mitigating pessimistic lower bounds of classic worst-case analysis via untrusted predictions~\cite{MitzenmacherV22,alps} has been successfully applied to various scheduling problems~\cite{PurohitSK18,LattanziLMV20,AzarLT21,AzarLT22,Im0QP21,LiX21online,LindermayrM22,AntoniadisGS22,DinitzILMV22portfolios}. While all these results concentrate on the uncertainty of online arrival and non-clairvoyance, Balkanski et al.~\cite{BalkanskiOSW22speedpredictions} consider a robust scheduling problem where machine speeds are only predicted and jobs have to be grouped to be scheduled together before knowing the true machine speeds; such problems without predictions were introduced in~\cite{EberleHMNSS21,SteinZ20}.
In contrast, in our model an algorithm will never learn about a job's true speed(s) before its completion and, further, the speeds might be job-dependent.

\section{Algorithms with Speed Predictions}\label{sec:speed-predictions}

In this section, we investigate the model with speed predictions. We first rule out any sublinear error-dependency.

\begin{theorem}\label{lemma:mu-lb}
  Any speed-oblivious algorithm with speed predictions has a competitive ratio of at least~$\Omega(\min\{\mu,m\})$ for minimizing the total (weighted) completion time, even if the algorithm is clairvoyant and machines are related.
\end{theorem}

\begin{proof}
Let~$\mu_1,\mu_2 \geq 1$ and~$\mu = \mu_1 \cdot \mu_2$. Consider an instance~$J=\{j\}$ with~$p_j = 2\mu$ and~$m \geq 2\mu$ machines such that~$\hs_{i} = \mu_1$ for all~$1\leq i\leq m$. The algorithm cannot distinguish the machines.
For the first~$2\mu-1$ machines~$i$ on which the algorithm processes~$j$, the adversary fixes~$s_{i}=1$. Thus, at time~$2\mu - 1$, the remaining processing requirement of~$j$ is at least~$2\mu - (2\mu-1) = 1$ and there exists a machine~$i'$ on which~$j$ has not been processed yet. Thus, the adversary can set~$s_{i'} = \mu$ and complete~$j$ on~$i'$ within two time units, implying a competitive ratio of at least~$\Omega(\min\{\mu,m\})$.  
\end{proof}

Observe that this construction already works for two machines when migration is forbidden.

\subsection{A Clairvoyant Algorithm}

We firstly present a novel migratory algorithm for the clairvoyant setting with known processing requirements for both the speed-aware setting as well as speed predictions.
Sequencing jobs by Smith's rule by non-increasing density~$\frac{w_j}{p_j}$ (aka Weighted-Shortest-Processing-Time, WSPT) is optimal on a single machine~\cite{smith1956various}. In the online setting with release dates, this policy is~$2$-competitive when applied preemptively on the available unfinished jobs~\cite{SchulzSk2002}. It can be extended to identical parallel machines~\cite{MegowS04}, by processing at any time the (at most)~$m$ jobs with highest densities. However, this approach is infeasible on unrelated machines, because jobs can have different densities on every machine. 

Inspired by the power of densities, we compute a subset of at most $m$ jobs that instead maximizes the total density, that is, the sum of the densities of the job-to-machine assignment. This can be done efficiently by computing at any time $t$ a matching $M_t$ between alive jobs $j\in J(t) = \{j \in J \mid r_j \leq t \leq C_j \}$ and machines $i\in I$ with edge weights $\hdense_{ij} = \frac{w_j \hs_{ij}}{p_{j}}$ using, e.g., the Hungarian algorithm~\cite{Kuhn55}. In the analysis, we crucially exploit the local optimality of any two matched job-machine pairs via exchange arguments. \begin{algorithm}[tb]
    \caption{Maximum Density}
    \label{alg:max-denses}
    \begin{algorithmic}[1]
      \REQUIRE time~$t$, speed (predictions)~$\{\hs_{ij}\}$ 
        \STATE Construct a complete bipartite graph~$G_t = I \cup J(t)$ where an edge $(i,j) \in I \times J(t)$ has a weight equal to $\hdense_{ij} = \frac{w_j \hs_{ij}}{p_{j}}$.
        \STATE Compute a maximum-weight matching~$M_t$ for~$G_t$.
        \STATE Schedule jobs to machines according to~$M_t$ at time $t$.
    \end{algorithmic}
\end{algorithm}

\begin{theorem}\label{thm:max-dense}
    \Cref{alg:max-denses} has a competitive ratio of at most~$8\mu$ for minimizing the total weighted completion time on unrelated machines with speed predictions.
\end{theorem}

This theorem implies immediately the following corollary for the speed-aware setting ($\mu = 1$).
\begin{corollary}
\Cref{alg:max-denses} has a competitive ratio of at most~$8$ for minimizing the total weighted completion time on unrelated machines in the speed-aware setting.
\end{corollary}

The remaining section is devoted to proof of \Cref{thm:max-dense}, which uses a dual-fitting argumentation.
To this end, we state the standard migratory linear programming relaxation for our objective function~\cite{SchulzS02}. In fact, we state a variant where the machines of an optimal solution run at a lower speed of~$\frac{1}{\alpha}$ for~$\alpha \geq 1$~\cite{ImKM18}.
\begin{alignat}{3}
    \text{min} \quad &\sum_{i \in I} \sum_{j \in J} \sum_{t \geq 0} w_{j} \cdot t \cdot \frac{x_{ijt} s_{ij}}{p_j} &\quad& \tag{$\text{LP}_\alpha$}\label{pf-lp} \\
    \text{s.t.} \quad& \sum_{i \in I} \sum_{t \geq 0} \frac{x_{ijt} s_{ij}}{p_{j}} \geq 1   &&\forall j \in J \notag \\
    & \sum_{j \in J} \alpha \cdot x_{ijt} \leq 1   &&\forall i \in I,t \geq 0 \notag \\
    & \sum_{i \in I} \alpha \cdot x_{ijt} \leq 1   &&\forall j \in J,t \geq r_j \notag \\
    & x_{ijt} \geq 0 &&\forall i \in I, j \in J,t\geq r_j \notag \\
    & x_{ijt} = 0 &&\forall i \in I, j \in J,t < r_j \notag
\end{alignat}

Let~$\opt_\alpha$ denote the optimal objective value in this restricted setting. 
The dual of \eqref{pf-lp} can be written as follows. (From now on we omit obvious set constraints in the notation for an improved readability.)
\begin{alignat}{3}
    \text{max} \quad &\sum_{j} \dualVa_j - \sum_{i,t} \dualVb_{it} - \sum_{j,t \geq r_j} \dualVc_{jt} \tag{$\text{DLP}_\alpha$}\label{pf-dual} \\
    \text{s.t.} \quad&\frac{\dualVa_j s_{ij}}{p_j} - \alpha  \dualVb_{it} - \alpha  \dualVc_{jt} \leq   w_j   \frac{s_{ij}  t}{p_j} \qquad &\forall i,j,t\geq r_j \notag\\
    &\dualVa_j, \dualVb_{it}, \dualVc_{jt'}  \geq 0 \qquad &\forall i,j,t \; \forall t'\geq r_j \notag
\end{alignat}

Fix an instance and the algorithm's schedule. 
Let~$\kappa \geq 1$ be a constant. We define for every machine~$i$ and any time~$t$ 
\[
\beta_{it} = \begin{cases} \hdense_{ij} & \text{ if } i \text{ is matched to } j \in J(t) \text{ in } M_t \\ 0 & \text{ otherwise,} \end{cases}
\]
and for every job~$j$ and any time~$t$
\[
\gamma_{jt} = \begin{cases} \hdense_{ij} & \text{ if } j \text{ is matched to } i \in I \text{ in } M_t \\ 0 & \text{ otherwise.} \end{cases}
\]

Consider the following values:
\begin{itemize}
  \item $\dualSa_j = w_j C_j$ for every job~$j$,
  \item $\dualSb_{it} = \frac{1}{\kappa} \sum_{t' \geq t} \beta_{it'}$ for every machine $i$ and time $t$, and
  \item $\dualSc_{jt} = \frac{1}{\kappa} \sum_{t' \geq t} \gamma_{jt'}$ for every job $j$ and time $t \geq r_j$.
\end{itemize}

We show in \Cref{lemma:maxdense-dual-feasible} that these values define a feasible solution for the dual problem~\eqref{pf-dual}, and that the corresponding dual objective value is at least a certain fraction of the algorithm's solution value (\Cref{lemma:maxdense-dual-objective}). Weak LP duality then implies~\Cref{thm:max-dense}. Let $\alg = \sum_j w_j C_j$.

\begin{lemma}\label{lemma:maxdense-dual-objective}
    $(1 - \frac{2\mu_1}{\kappa})  \alg \leq \sum_{j} \dualSa_j - \sum_{i,t} \dualSb_{it} - \sum_{j,t \geq r_j} \dualSc_{jt}$
\end{lemma}

In the following, let~$U_t$ be the set of unfinished jobs at time~$t$, i.e., all jobs~$j$ with~$t \leq C_j$, and let~$W_t = \sum_{j \in U_t} w_j$. 

\begin{proof}
Fix a time~$t$ and a job~$j$. If~$j \in U_t$, let~$i^j_1,\ldots,i^j_{z(j)}$ be the sequence of individual machine assignments of~$j$ between time~$t$ and~$C_j$. Let $\hdense(i,j) := \hdense_{ij}$. Note that 
    \[
        \sum_{\ell=1}^{z(j)} \hdense(i^j_\ell, j) = \sum_{\ell=1}^{z(j)} \hs_{i^j_\ell,j} \frac{w_j}{p_j} \leq \mu_1 \sum_{\ell=1}^{z(j)}  s_{i^j_\ell,j} \frac{w_j}{p_j} \leq \mu_1 w_j.
    \]
    Therefore,~$\sum_{i} \dualSb_{it} = \frac{1}{\kappa} \sum_{j \in U_t} \sum_{\ell=1}^{z(j)} \hdense(i^j_\ell, j) \leq \frac{\mu_1}{\kappa} W_t$.
    Similarly,~$\dualSc_{jt} = \frac{1}{\kappa} \sum_{\ell=1}^{z(j)} \hdense(i^j_\ell, j) \leq  \frac{\mu_1}{\kappa} w_j$. If~$j \in J \setminus U_t$, then,~$\dualSc_{jt} = 0$. Hence,~$\sum_{j \in J} \dualSc_{jt} \leq \frac{\mu_1}{\kappa} W_t$.
    Finally, we conclude~$\sum_{i,t} \dualSb_{it} \leq \frac{\mu_1}{\kappa} \alg$ and~$\sum_{j,t \geq r_j} \dualSc_{jt} \leq \frac{\mu_1}{\kappa} \alg$.  
\end{proof}

  \begin{lemma}\label{lemma:maxdense-dual-feasible}
    Assigning $\dualVa_j = \dualSa_j$, $\dualVb_{it} = \dualSb_{it}$ and $\dualVc_{jt} = \dualSc_{jt}$ is feasible for~\eqref{pf-dual} if $\alpha = \mu_2 \kappa$.
  \end{lemma}

  \begin{proof}
   First note that the dual assignment is non-negative. Let~$i \in I, j \in J$ and~$t \geq r_j$. The definition of~$\dualSa_j$ yields
    \(
      \dualSa_j \frac{s_{ij}}{p_j} - w_j t \frac{s_{ij}}{p_j} \leq \sum_{t'=t}^{C_j} \frac{w_j s_{ij}}{p_j}.
    \)
    By using the fact that~$\frac{w_j s_{ij}}{p_j} \leq \mu_2 \frac{w_j \hs_{ij}}{p_j}$, the definitions of~$\dualSb_{it}$ and~$\dualSc_{jt}$, and the value of~$\alpha$,
    it remains to validate for every~$t \leq t' \leq C_j$ that~$\frac{w_j \hs_{ij}}{p_j} = \hdense_{ij} \leq \beta_{it'} + \gamma_{jt'}$. We distinguish five cases:
    \begin{enumerate}[(i)]
      \item If~$(i,j) \in M_{t'}$, then~$\hdense_{ij} = \beta_{it'} = \gamma_{jt'}$.
      \item If~$(i, j') \in M_{t'}$ and~$(i',j) \in M_{t'}$ s.t.~$i' \neq i$ (and thus~$j' \neq j$), we know by the optimality of~$M_{t'}$ that
      \begin{align*}
        \hdense_{ij} \leq \hdense_{ij} + \hdense_{i'j'} \leq \hdense_{i'j} + \hdense_{ij'} = \gamma_{jt'} + \beta_{it'}.
        \end{align*}
      \item If~$(i',j) \in M_{t'}$ and~$i$ is not matched in~$M_{t'}$, we conclude~\(
         \hdense_{ij} \leq \hdense_{i'j} = \gamma_{jt'}.
      \)

      \item If~$(i,j') \in M_{t'}$ and~$j$ is not matched in~$M_{t'}$, we conclude~\(
         \hdense_{ij} \leq \hdense_{ij'} = \beta_{it'}.
      \)
      \item The case where~$\hs_{ij} > 0, w_j > 0$, but both~$i$ and~$j$ are unmatched in~$M_{t'}$ contradicts the optimality of~$M_{t'}$, as~$t' \leq C_j$. Else holds $\hdense_{ij} = 0$, and we conclude since the right side of the inequality is non-negative. \endproof
    \end{enumerate}
  \end{proof}

  \begin{proof}[Proof of~\Cref{thm:max-dense}] 
  Weak LP duality implies that the optimal objective value of~\eqref{pf-dual} is greater or equal to the optimal objective value of~\eqref{pf-lp}. Being the objective value of a relaxation, the latter is a lower bound on $\opt_\alpha$, which in turn is at most $\alpha \opt$ by scaling completion times, where~$\opt$ denotes the optimal objective value of the original problem.
This implies via
    \Cref{lemma:maxdense-dual-objective} and \Cref{lemma:maxdense-dual-feasible} \begin{align*}
        \mu_2 \kappa \cdot \opt \geq \opt_{\mu_2 \kappa} \geq \sum_{j} \dualSa_j - \sum_{i,t} \dualSb_{it} - \sum_{j,t \geq r_j} \dualSc_{jt} \geq \left( 1 - \frac{2 \mu_1}{\kappa} \right) \cdot \alg.
     \end{align*}
Choosing $\kappa = 4\mu_1$, we conclude $\alg \leq 8 \mu \cdot \opt$.
\end{proof}

\subsection{A Clairvoyant Non-Preemptive Algorithm}

\begin{algorithm}[tb]
  \caption{Greedy WSPT}\label{alg:min-increase}
  \begin{algorithmic}
    \REQUIRE speed predictions $\{\hs_{ij}\}$
    \FUNCTION{UponJobArrival(job $j$)}{
      \STATE Assign job $j$ to machine $g(j) = \argmin_{i \in I} \hQ_{ij}$.
    }
    \ENDFUNCTION

    \FUNCTION{UponMachineIdle(machine $i$, time $t$)}{
          \STATE Start processing the job $j$ with largest $\hdense_{ij}$ among all alive jobs assigned to $i$ which satisfy $\hr_{ij} \leq t$.
        }
    \ENDFUNCTION

\end{algorithmic}
\end{algorithm}

In many applications, job migration or preemption are not possible.
In this section, we show that the non-preemptive Greedy~WSPT algorithm by~\cite{GuptaMUX20} achieves an error-dependent competitive ratio when using predicted speeds to make decisions (\Cref{alg:min-increase}).
The intuition of this algorithm is to greedily assign arriving jobs to machines, where they are then scheduled in WSPT order, i.e., on machine $i$ by non-decreasing $\frac{w_j s_{ij}}{p_j}$. The greedy job-to-machine assignment intuitively minimizes the increase of the objective value that scheduling the job on a machine incurs in the current state. Additionally, the execution of job $j$ is delayed depending on its processing time $\frac{p_j}{s_{ij}}$ on the assigned machine $i$. This is necessary due to simple lower bounds in the non-preemptive setting~\cite{LuSS03}.

To make this precise, for every~$j \in J$, let~$M_i(j)$ be the set of jobs, excluding job~$j$, which are
assigned to machine~$i$ at time~$r_j$, but have not been started yet. As this definition is ambiguous if there are two jobs~$j$ and~$j'$ with~$r_{j} = r_{j'}$ being assigned to~$i$, we assume that we assign them in the order of
their index.
For all machines~$i$, jobs~$j$ and a constant~$\theta > 0$, which we will set $\theta = \frac{2}{3}$, we define~$\hr_{ij} = \max\{r_j, \theta \frac{p_j}{\hs_{ij}}\}$ and~$\hQ_{ij}$ as
\[
	w_{j} \Bigg( \hr_{ij} + \frac{\hr_{ij}}{\theta} + \frac{p_j}{\hs_{ij}} + \sum_{\substack{j' \in M_i(j) \\ \hdense_{ij'} \geq \hdense_{ij}}} \frac{p_{j'}}{\hs_{ij'}} \Bigg) + \frac{p_j}{\hs_{ij}} \sum_{\substack{j' \in M_i(j) \\ \hdense_{ij'} < \hdense_{ij}}} w_{j'}.
\]
We prove in \Cref{app:min-increase} the following theorem.

\begin{restatable}{theorem}{thmGreedyWSPT}\label{theorem:minincrease}
  \Cref{alg:min-increase} has a competitive ratio of at most $\frac{368}{51}\mu^2 < 7.216 \mu^2$ for minimizing the total weighted completion time on unrelated machines with speed predictions.
\end{restatable}

\subsection{A Non-Clairvoyant Algorithm}

In the non-clairvoyant setting, 
any constant-competitive algorithm for minimizing the total completion time on unrelated machines has to migrate and preempt jobs~\cite{MotwaniPT94,GuptaIKMP12}.
Since such algorithms cannot compute densities, a common strategy is to run all jobs simultaneously at a rate proportional to
their weight~\cite{MotwaniPT94,KimC03a}. 
On unrelated machines with job-dependent speeds, the Proportional Fairness Algorithm (PF) develops this idea further by respecting job-dependent speeds~\cite{ImKM18}.
It is known that PF has a competitive ratio of at most~$128$ for minimizing the total weighted completion time~\cite{ImKM18}.
In the following, we show that PF has a linear error-dependency in $\mu$ when computing rates via predicted speeds. 
As a byproduct, we slightly improve the upper bound on the speed-aware competitive ratio of PF via optimized duals to~$108$.

\begin{algorithm}[tb]
  \caption{Proportional Fairness}
  \label{alg:pf}
  \begin{algorithmic}
    \REQUIRE time $t$, speed predictions $\{\hs_{ij}\}$
      \STATE Use solution $\{y_{ijt}\}_{i,j}$ of \eqref{pf-convex} as rates at time $t$. \end{algorithmic}
\end{algorithm}

\begin{restatable}{theorem}{thmPF}\label{thm:pf}
  \Cref{alg:pf} has a competitive ratio of at most $108\mu$ for minimizing the total weighted completion time on unrelated machines with predicted speeds.
\end{restatable}

At every time $t$, \Cref{alg:pf} schedules jobs $J(t)$ with rates computed via the following convex program \eqref{pf-convex} with variables $\hy_{ijt}$ for every machine $i$ and job $j \in J(t)$.

\begin{alignat}{3}
  \text{max} \quad &\sum_{j \in J(t)} w_j \log\left(\sum_{i \in I} \hs_{ij} \hy_{ijt} \right) \tag{$\text{CP}_t$} \label{pf-convex}\\
  \text{s.t.} \quad& \sum_{j \in J(t)} \hy_{ijt} \leq 1   &&\forall i \in I \notag \\
  & \sum_{i \in I} \hy_{ijt} \leq 1   &&\forall j \in J(t) \notag \\
  & \hy_{ijt} \geq 0 &&\forall i \in I, j \in J(t) \notag 
\end{alignat}

We now give an overview over the proof of \Cref{thm:pf} and defer further details to~\Cref{app:pf}.

Fix an instance and PF's schedule.
Let~$\kappa \geq 1$ and~$0 < \lambda < 1$ be constants which we fix later.
In the following, we assume by scaling that all weights are integers.
For every time~$t$, let~$Z^t$ be the sorted (ascending) list of length~$W_t$ composed of $w_j$ copies of~$\frac{q_{jt}}{p_j}$ for every~$j \in U_t$. We define~$\zeta_t$ as the value at the index~$\floor{\lambda W_t}$ in~$Z^t$.
Let~$\{\eta_{it}\}_{i,t}$ and~$\{\theta_{jt}\}_{j \in J(t),t}$ be the KKT multipliers of the first two constraint sets of the optimal solution $\{y_{ijt}\}_{i,j}$. Let~$\ind[\varphi]$ be the indicator variable of the formula~$\varphi$, and consider the following duals:
\begin{itemize}
    \item $\dualSa_j = \sum_{t' = 0}^{C_j} w_j \cdot \ind \left[ \frac{q_{jt'}}{p_j} \leq \zeta_{t'} \right]$ for every job $j$,
    \item $\dualSb_{it} = \frac{1}{\kappa} \sum_{t' \geq t}   \zeta_{t'}  \eta_{it'}$ for every machine $i$ and time $t$, and
    \item $\dualSc_{jt} = \frac{1}{\kappa} \sum_{t' = t}^{C_j}   \zeta_{t'}  \theta_{jt'}$ for every job $j$ and time $t \geq r_j$.
\end{itemize}

We show that this assignment has an objective value which lower bounds a fraction of PF's objective value, and that it is feasible for \eqref{pf-dual} for some values of $\alpha$.

\begin{restatable}{lem}{pfDualObj}\label{lemma:pf-dual-objective}
    $(\lambda - \frac{4}{(1-\lambda)\kappa}) \alg \leq \sum_{j} \dualSa_j - \sum_{i,t} \dualSb_{it} - \sum_{j,t \geq r_j} \dualSc_{jt}$
\end{restatable}

\begin{restatable}{lem}{pfDualFeasibility}\label{lemma:pf-dual-feasibility}
    Assigning $\dualVa_j = \dualSa_j$, $\dualVb_{it} = \dualSb_{it}$ and $\dualVc_{jt} = \dualSc_{jt}$ is feasible for \eqref{pf-dual} if $\alpha = \kappa \mu$.
\end{restatable}

\begin{proof}[Proof of \Cref{thm:pf}]
  Weak duality, \Cref{lemma:pf-dual-objective} and \Cref{lemma:pf-dual-feasibility} imply
  \begin{align*}
      \kappa \mu \cdot \opt \geq \opt_{\kappa \mu}
      \geq  \sum_{j} \dualSa_j - \sum_{i,t} \dualSb_{it} - \sum_{j,t \geq r_j} \dualSc_{jt} \geq \left(\lambda -  \frac{4}{(1-\lambda)\kappa}\right) \cdot \alg.
  \end{align*}
  Setting $\kappa = 36$ and $\lambda = \frac{2}{3}$ implies $\alg \leq 108 \mu \cdot \opt$.
\end{proof}

\section{Algorithms for Speed-Ordered Machines}\label{sec:speed-ordered}

This section contains our results on speed-ordered machines. In the first subsection, we present a clairvoyant algorithm, and in the second subsection a non-clairvoyant algorithm. But first, we observe that in this model migration is necessary for speed-oblivious algorithms.

\begin{restatable}{theorem}{thmSOmigration}\label{thm:speed-order-migration}
Any non-migratory speed-oblivious algorithm has a competitive ratio of at least $\Omega(m)$ for minimizing the total completion time on $m$ speed-ordered machines, even if it is clairvoyant and the machines are related.
\end{restatable}

\begin{proof}
  Consider the execution of some algorithm on an instance of~$n$ jobs with unit-weights and with processing requirements equal to $n^2 m$ and $s_1 = n^2m$. If at some point in time, the algorithm starts a job on machines~$2,\ldots,m$, the adversary sets~$s_2 = \ldots = s_m = 1$ to enforce an objective value of at least~$\Omega(n^2m)$, while scheduling all jobs on the first machine gives an objective value of at most~$\bigO(n^2)$.
  If this does not happen, the algorithm must have scheduled all jobs on the first machine. But then the adversary sets~$s_2 = \ldots = s_m = n^2m$ and achieves an objective value of~$\bigO(\frac{n^2}{m})$ by distributing the jobs evenly to all machines, while the algorithm has an objective value of~$\Omega(n^2)$.
  \end{proof}

\subsection{A Clairvoyant Algorithm}

\begin{algorithm}[tb]
  \caption{Maximum Density for speed-ordered machines}
  \label{alg:max-denses-speed-order}
  \begin{algorithmic}[1]
    \REQUIRE time $t$, speed-ordered machines $s_1 \geq \ldots \geq s_m$
      \STATE $\sigma_t \gets$ order of $J(t)$ with non-increasing $\frac{w_j}{p_j}$.
      \STATE $M_t = \{(k, \sigma_t(k))\}_{k \in [\ell]}$ where $\ell = \min\{m,\abs{J(t)}\}$
      \STATE Schedule jobs to machines according to $M_t$ at time $t$.
  \end{algorithmic}
\end{algorithm}

Our clairvoyant algorithm for speed-ordered related machines is motivated by the following observation. If the machines are related and speed-ordered, \Cref{alg:max-denses}, given correct speed predictions, will assign jobs by non-increasing order of~$\frac{w_j}{p_j}$ to machines in speed order, because this clearly maximizes the total scheduled density, i.e., sum of assigned~$\frac{w_j s_i}{p_j}$. 
\Cref{alg:max-denses-speed-order} can therefore emulate this schedule of maximum density \emph{without} having to compute a maximum matching, and thus does not require (predicted) speeds.
These observations also suggest that the analysis must be similar. Indeed, we can use a similar dual-fitting as for \Cref{thm:max-dense}
to prove the following theorem. We mainly present new ideas for proving the dual feasibility. Note that this observation does not hold for unrelated machines.

\begin{restatable}{theorem}{thmMaxDenseSO}\label{thm:max-dense-speed-order}
\Cref{alg:max-denses-speed-order} has a competitive ratio of at most~$8$ for minimizing the total weighted completion time on speed-ordered related machines.
\end{restatable}
  
We use a dual-fitting analysis based on \eqref{pf-dual} to prove this theorem. Fix an instance and the algorithm's schedule, and observe that the algorithm ensures at every time $t$ that $M_t$ is a matching between alive jobs and machines. Recall that for related machines, $s_i = s_{ij}$ for every job $j$ and every machine $i$.
Let $\kappa \geq 1$ be a constant. We define for every machine $i$ and any time $t$ 
\[
\beta_{it} = \begin{cases} \frac{w_j s_i}{p_j} & \text{ if } i \text{ is matched to } j \in J(t) \text{ in } M_t \\ 0 & \text{ otherwise,} \end{cases}
\]
and for every job~$j$ and any time~$t$
\[
\gamma_{jt} = \begin{cases} \frac{w_j s_i}{p_j} & \text{ if } j \text{ is matched to } i \in I \text{ in } M_t \\ 0 & \text{ otherwise.} \end{cases}
\]

Using these values, we have the following dual assignment:
\begin{itemize}
\item $\dualSa_j = w_j C_j$ for every job~$j$,
\item $\dualSb_{it} = \frac{1}{\kappa} \sum_{t' \geq t} \beta_{it'}$ for every machine $i$ and time $t$, and
\item $\dualSc_{jt} = \frac{1}{\kappa} \sum_{t' \geq t} \gamma_{jt'}$ for every job $j$ and time $t \geq r_j$.
\end{itemize}

We first observe that the dual objective of this assignment is close to algorithm's objective. The proof works analogous to the proof of~\Cref{lemma:maxdense-dual-objective}.

\begin{lemma}\label{lemma:maxdense-speed-order-dual-objective}
$(1 - \frac{2}{\kappa})  \alg \leq \sum_{j} \dualSa_j - \sum_{i,t} \dualSb_{it} - \sum_{j,t \geq r_j} \dualSc_{jt}$
\end{lemma}

  \begin{lemma}\label{lemma:maxdense-speed-order-dual-feasible}
    Assigning $\dualVa_j = \dualSa_j$, $\dualVb_{it} = \dualSb_{it}$ and $\dualVc_{jt} = \dualSc_{jt}$ is feasible for~\eqref{pf-dual} if $\alpha = \kappa$ and $s_i = s_{ij}$ for every job $j$ and every machine $i$.
  \end{lemma}

  \begin{proof}
   Since the dual assignment is clearly non-negative, we now show that it satisfies the dual constraint.  Let $i \in I, j \in J$ and $t \geq r_j$. We first observe that
    \begin{align*}
      \dualSa_j \frac{s_{i}}{p_j} - w_j t \frac{s_{i}}{p_j} \leq \sum_{t'=t}^{C_j} \frac{w_j s_{i}}{p_j}.
    \end{align*}

    Using $\alpha = \kappa$, it remains to validate for every $t \leq t' \leq C_j$ that $\frac{w_j s_{i}}{p_j} \leq \beta_{it'} + \gamma_{jt'}$. We distinguish five cases:
    \begin{enumerate}[(i)]
      \item If $(i,j) \in M_{t'}$, then $\frac{w_j s_{i}}{p_j} = \beta_{it'} = \gamma_{jt'}$.
      \item If $(i, j') \in M_{t'}$ and $(i',j) \in M_{t'}$ s.t. $i \neq i'$, we have two cases. If $i < i'$, it must be that~$\sigma_{t'}(j') < \sigma_{t'}(j)$ and, thus,~$\frac{w_{j'}}{p_{j'}} \geq \frac{w_j}{p_j}$. But then,~$\frac{w_j s_{i}}{p_j} \leq \frac{w_{j'} s_i}{p_{j'}}$. Otherwise, that is,~$i > i'$, we know by the speed order that~$s_{i} \leq s_{i'}$, and, thus,~$\frac{w_j s_{i}}{p_j} \leq \frac{w_{j} s_{i'}}{p_{j}}$. Put together,
      \[
        \frac{w_j s_{i}}{p_j} \leq \frac{w_{j'} s_i}{p_{j'}} + \frac{w_{j} s_{i'}}{p_{j}} =  \beta_{it'} + \gamma_{jt'}.
      \]
      \item If $(i',j) \in M_{t'}$ and $i$ is not matched in $M_{t'}$, it follows $i' < i$, which gives
      \(
         \frac{w_j s_{i}}{p_j} \leq \frac{w_j s_{i'}}{p_j} = \gamma_{jt'}.
      \)

      \item If $(i,j') \in M_{t'}$ and $j$ is not matched in $M_{t'}$, it follows $\sigma_{t'}(j') < \sigma_{t'}(j)$, and hence $\frac{w_j}{p_j} \leq \frac{w_{j'}}{p_{j'}}$. This immediately concludes
      \(
         \frac{w_j s_{i}}{p_j} \leq \frac{w_{j'} s_{i}}{p_{j'}} = \beta_{it'}.
      \)
      \item The case where both $i$ and $j$ are unmatched in $M_{t'}$ contradicts the definition of $M_{t'}$ in \Cref{alg:max-denses-speed-order}.
    \end{enumerate}
  \end{proof}

  \begin{proof}[Proof of~\Cref{thm:max-dense-speed-order}]
    Weak duality, \Cref{lemma:maxdense-speed-order-dual-feasible} and \Cref{lemma:maxdense-speed-order-dual-objective} imply
     \begin{align*}
        \kappa \cdot \opt \geq \opt_{\kappa} \geq \sum_{j} \dualSa_j - \sum_{i,t} \dualSb_{it} - \sum_{j,t \geq r_j} \dualSc_{jt} \geq \left( 1 - \frac{2}{\kappa} \right) \cdot \alg.
     \end{align*}
     Using $\kappa = 4$ concludes
     \(
      \alg \leq \frac{\kappa}{1 - 2  / \kappa} \cdot \opt = 8 \cdot \opt.
     \)
  \end{proof}

We finally observe that \Cref{alg:max-denses-speed-order} indeed cannot achieve a good competitive ratio if speeds are job-dependent.

\begin{lemma}
  \cref{alg:max-denses-speed-order} has a competitive ratio of at least $\Omega(n)$ for minimizing the total weighted completion time on speed-ordered unrelated machines, even on two machines and if $w_j = 1$ for all jobs $j$.
\end{lemma}

\begin{proof}
  Let~$0 < \epsilon < 1$. Consider an instance composed of~$n$ jobs and~$2$ machines, where~$w_j = 1$ for all jobs~$j$,~$p_1 = 1$ and~$p_{j} = 1 + \epsilon$ for all~$2 \leq j \leq n$. The processing speeds are given by~$s_{11} = s_{21} = \epsilon$, and~$s_{1j} = 1$ and~$s_{2j} = \epsilon$ for all~$2 \leq j \leq n$. Note that the machines are speed-ordered. \Cref{alg:max-denses-speed-order} completes at time
~$\frac{1}{\epsilon}$ job~$1$ on machine~$1$ before any other job. Thus,~$\alg \geq \frac{n}{\epsilon}$. Another solution is to schedule jobs~$2,\ldots,n$ on machine~$1$, and job~$1$ on machine~$2$, giving an objective of at most~$n^2 + \frac{1}{\epsilon}$. For~$\epsilon < n^{-2}$, this concludes that~$\frac{\alg}{\opt} \geq \Omega(n)$.
\end{proof}

\subsection{A Non-Clairvoyant Algorithm}

The non-clairvoyant setting is more difficult. This is because the schedules of speed-aware algorithms, such as PF, are not as easy to describe, as it was the case for clairvoyant algorithms. However, for unit weights, related machines and many alive jobs, i.e., $\abs{J(t)} \geq m$, one solution of \eqref{pf-convex} is to schedule all jobs on all machines with the same rate, i.e., do Round Robin on every machine. We can describe this schedule without knowing anything about the speeds. However, in the few-job regime, i.e., $\abs{J(t)} < m$, this approach violates the packing constraints of the jobs, i.e., $\sum_i y_{ijt} > 1$. This is where the speed order comes into play: 
we partition a job's available rate only to the $\abs{J(t)}$ fastest machines.
For the final algorithm (\Cref{alg:rr-speed-order}), we prove below a guarantee for unrelated machines, and a constant upper bound for related machines in \Cref{app:rr-speed-ordered-related}.

\begin{algorithm}[tb]
  \caption{Round Robin for speed-ordered machines}
  \label{alg:rr-speed-order}
  \begin{algorithmic}
    \REQUIRE time $t$, speed-ordered machines $s_{1j} \geq \ldots \geq s_{mj}$
      \STATE Use rates $y_{ijt} = \abs{J(t)}^{-1} \cdot \ind\left[i \leq \abs{J(t)}\right]$ at time $t$.
\end{algorithmic}
\end{algorithm}

\begin{restatable}{theorem}{thmRoundRobinSOunrelated}\label{thm:unrelated-speed-order}
\Cref{alg:rr-speed-order} has a competitive ratio of at most~$\bigO(\log(\min\{n,m\}))$ for minimizing the total completion time on speed-ordered unrelated machines.
\end{restatable}

We prove \Cref{thm:unrelated-speed-order} via dual-fitting based on~\eqref{pf-dual}, where~$w_j = 1$ for every job~$j$. 
Fix an instance and the algorithm's schedule. For every time~$t$, we write~$m_t = \min\{m, \abs{J(t)}\}$, and we define~\(
  \beta_{it} = \frac{1}{i} \cdot \abs{J(t)} \cdot \ind\left[ i \leq \abs{J(t)}  \right]
\)
for every machine~$i$,
and~\(
  \gamma_{jt} = \ind\left[ j \in J(t) \right]
\) for every job~$j$.

Let~$\kappa = \Theta(\log(\min\{n,m\}))$. Intuitively, this factor upper bounds~$\sum_{i = 1}^{m_{t}} \frac{1}{i}$, which will be necessary when handling~$\sum_{i} \beta_{i t}$. For related machines, we can alter the definition of $\beta_{it}$ and thus have a constant~$\kappa$, which eventually implies a constant upper bound on the competitive ratio.

For every time~$t$, consider the sorted (ascending) list~$Z^t$ composed of values~$\frac{q_{jt}}{p_j}$ for every~$j \in U_t$. 
We define~$\zeta_t$ as the value at the index~$\floor{\frac{1}{2} \abs{U_t}}$ in~$Z^t$.
Consider the following duals:
\begin{itemize}
  \item~$\dualSa_j = \sum_{t' = 0}^{C_j} \ind \left[ \frac{q_{jt'}}{p_j} \leq \zeta_{t'} \right]$ for every job~$j$, 
  \item $\dualSb_{it} = \frac{1}{\kappa} \sum_{t' \geq t} \beta_{it'} \zeta_{t'}$ for every machine~$i$  and time $t$, and
  \item $\dualSc_{jt} = \frac{1}{\kappa} \sum_{t' \geq t} \gamma_{jt'} \zeta_{t'}$ for every job~$j$ and time $t \geq r_j$.
\end{itemize}

We prove the following bound on $\alg$ in \Cref{app:rr-speed-ordered-unrelated}.

\begin{lemma}\label{lemma:rr-speed-order-dual-objective-unrelated}
  $\Omega(1) \cdot \alg \leq \sum_{j} \dualSa_j - \sum_{i,t} \dualSb_{it} - \sum_{j,t \geq r_j} \dualSc_{jt}$
\end{lemma}

This lemma, weak LP duality, and the feasibility of the crafted duals (\Cref{lemma:rr-speed-order-dual-feasibility-unrelated}) imply \Cref{thm:unrelated-speed-order} for $\alpha = \kappa$.

\begin{lemma}\label{lemma:rr-speed-order-dual-feasibility-unrelated}
  Assigning $\dualVa_j = \dualSa_j$, $\dualVb_{it} = \dualSb_{it}$ and $\dualVc_{jt} = \dualSc_{jt}$ is feasible for~\eqref{pf-dual} if $\alpha = \kappa$.
\end{lemma}

\begin{proof}
  First observe that the dual assignment is non-negative. Let~$i \in I, j \in J$ and~$t \geq r_j$. Since the rates of \Cref{alg:rr-speed-order} imply \( q_{jt} = \sum_{\ell = 1}^{m_t} \frac{s_{\ell j}}{\abs{J(t)}} \), we have
  \begin{align}
    \frac{\dualSa_j s_{ij}}{p_j} - \frac{ s_{ij} \cdot t}{p_j} &\leq \sum_{t' = t}^{C_j} \frac{ s_{ij}}{p_j} \cdot \ind\left[ \frac{q_{jt'}}{p_j} \leq \zeta_{t'} \right] = \sum_{t' = t}^{C_j} \frac{ s_{ij}}{q_{jt'}} \cdot \frac{q_{jt'}}{p_j} \cdot \ind\left[ \frac{q_{jt'}}{p_j} \leq \zeta_{t'} \right] \leq \sum_{t' = t}^{C_j} \frac{ s_{ij}}{\sum_{\ell = 1}^{m_{t'}} \frac{s_{\ell j} }{\abs{J(t')}}} \cdot \zeta_{t'}. \label{eq:speed-ordered-rr-unrelated}
  \end{align}

  Consider any time~$t'$ with~$t \leq t' \leq C_j$. If~$i \leq \abs{J(t')}$, by the speed order, $\sum_{\ell = 1}^{m_{t'}} s_{\ell j} \geq \sum_{\ell = 1}^{i} s_{\ell j} \geq i \cdot s_{ij}$, and thus
  \[
    \frac{s_{ij}}{\sum_{\ell = 1}^{m_{t'}} s_{\ell j}} \cdot \abs{J(t')} \cdot \zeta_{t'} \leq \frac{1}{i} \cdot \abs{J(t')} \cdot \zeta_{t'} = \beta_{it'} \cdot \zeta_{t'}.
  \]

  Otherwise, that is,~$i > \abs{J(t')}$, we conclude by the speed order, $\sum_{\ell = 1}^{m_{t'}} s_{\ell j} \geq \sum_{\ell = 1}^{\abs{J(t')}} s_{\ell j} \geq \abs{J(t')} \cdot s_{ij}$. Therefore,
  \[
    \frac{s_{ij}}{\sum_{\ell = 1}^{m_{t'}} s_{\ell j}} \cdot \abs{J(t')} \cdot \zeta_{t'} \leq \frac{\abs{J(t')}}{\abs{J(t')}} \cdot \zeta_{t'} =  \gamma_{jt'} \cdot \zeta_{t'},
  \]
  because~$t' \leq C_j$. 
  Put together, \eqref{eq:speed-ordered-rr-unrelated} is at most
  \begin{align*}
\sum_{t' = t}^{C_j} \beta_{it'}\zeta_{t'} + \sum_{t' = t}^{C_j} \gamma_{jt'}\zeta_{t'} \leq \kappa (\dualSb_{it} + \dualSc_{jt}),
\end{align*}
which verifies the dual constraint.
\end{proof}

The following lower bound shows that \Cref{thm:unrelated-speed-order} is tight up to a constant factor.

\begin{lemma}
\Cref{alg:rr-speed-order} has a competitive ratio of at least~$\Omega(\log(\min\{n,m\}))$ for minimizing the total completion time on speed-ordered unrelated machines, even if all jobs are released at time zero and processing speeds are exclusively from~$\{0,1\}$.
\end{lemma}

\begin{proof}
We construct instances on which \Cref{alg:rr-speed-order} completes all jobs simultaneously, whereas dedicating disjoint groups of machines to groups of jobs completes them at staggered times.

Fix an integer~$K \geq 2$, and set~$m = 2^{K-1}$ and~$n = Km$. Let~$I = [m]$, and partition the job set~$J$ into classes~$J_0,\ldots,J_{K-1}$ of~$m$ jobs each. Every job~$j \in J_k$ has
\begin{compactitem}
  \item release date~$r_j = 0$,
  \item weight~$w_j = 1$,
  \item processing requirement~$p_j = 2^k$, and
  \item speed~$s_{ij} = \ind\left[i \leq 2^k\right]$ on every machine~$i \in I$.
\end{compactitem}
Thus a job of~$J_k$ can be processed only on the~$2^k$ fastest machines, and its processing requirement equals this number of machines. The machines are speed-ordered, since~$s_{1j} \geq \ldots \geq s_{mj}$ for every job~$j$.

We first compute~$\alg$. While all~$n$ jobs are alive, we have~$\abs{J(t)} = n \geq m$, so \Cref{alg:rr-speed-order} assigns every job rate~$1/n$ on every machine. A job~$j \in J_k$ then progresses at rate~$q_{jt} = \sum_{i \in I} s_{ij} y_{ijt} = 2^k/n = p_j/n$. Hence every job~$j$ has received progress~$t p_j/n < p_j$ by any time~$t < n$, so all jobs remain alive until time~$n$ and complete exactly at time~$n$. This gives
\[
\alg = \sum_{j \in J} C_j = n^2 = K^2m^2\,.
\]

To bound~$\opt$, dedicate machine~$1$ to~$J_0$ and, for each~$k \geq 1$, the block~$B_k = \{2^{k-1}+1,\ldots,2^k\}$ of~$2^{k-1}$ machines to~$J_k$. These blocks are disjoint, and every job of~$J_k$ has speed~$1$ on every machine of~$B_k$, because~$B_k \subseteq [2^k]$. Distribute each class evenly across its machines, which is possible because~$2^{k-1}$ divides~$m$, and process the jobs on each machine sequentially from time zero.

Class~$J_0$ contributes~$\sum_{\ell=1}^{m} \ell = m(m+1)/2$ to the total completion time. For~$k \geq 1$, each of the~$2^{k-1}$ machines in~$B_k$ processes~$m/2^{k-1}$ jobs of length~$2^k$, and the~$\ell$-th of them completes at time~$\ell \cdot 2^k$. The contribution of~$J_k$ is therefore
\[
2^{k-1} \sum_{\ell=1}^{m/2^{k-1}} \ell \cdot 2^k = m^2 + 2^{k-1}m\,.
\]
Summing over all classes and using~$\sum_{k=1}^{K-1} 2^{k-1} = m - 1$, this schedule shows
\[
\opt \leq \frac{m(m+1)}{2} + \sum_{k=1}^{K-1}\left(m^2 + 2^{k-1}m\right) = \left(K+\frac12\right)m^2-\frac m2\,.
\]
Consequently,
\[
\frac{\alg}{\opt} \geq \frac{K^2}{K+\frac12-\frac1{2m}} > \frac{K^2}{K+\frac12} > K-\frac12 = \log_2 m+\frac12\,.
\]
Since~$\min\{n,m\}=m$ and~$K$ is arbitrary, the competitive ratio is at least~$\Omega(\log(\min\{n,m\}))$.
\end{proof}

\section{Experimental Evaluation}\label{sec:experiments}

\begin{figure*}
	\centering
    \includegraphics{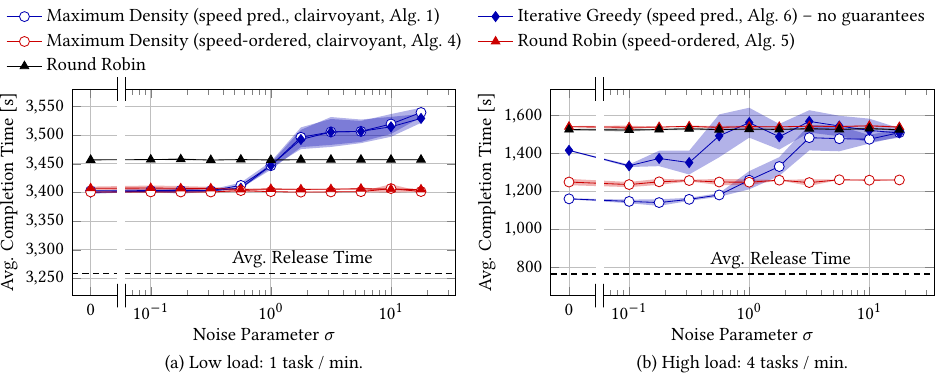}
	\caption{Real experiments on a \hikey board. The experiments are each repeated 3 times with the same workload but different random noise for speed predictions. Shaded areas show the standard deviation.}
	\label{fig:results_hikey}
\end{figure*}

\paragraph{Setup}

We perform experiments on real hardware running representative jobs, which enables us to perform a realistic evaluation.
The setup uses a \hikey board~\cite{hikey970} with a \emph{Kirin~970} Arm big.LITTLE \ac{SoC} featuring 4~\bigc cores and 4~\littlec cores, running Android~8.0.
This is a representative smartphone platform.
The \bigc cores always offer a higher performance than the \littlec cores (speed-ordered) because they support out-of-order execution at higher frequency and larger caches (see also \Cref{fig:characterization}, all speedups are $>1$).
Our workload comprises 100 randomly selected single-threaded jobs from the well-established \parsec~\cite{parsec3}, \splash~\cite{splash3}, and Polybench~\cite{polybench} benchmark suites.
These benchmarks represent various use cases from video transcoding, rendering, compression,
etc.
The arrival times are drawn from a Poisson distribution with varying rate parameter to study different system loads.
We characterized all jobs offline to get accurate speed~$s_{ij}$ and job volume~$p_j$ values.
Speed predictions are created with controllable error by~$\hs_{ij}=s_{ij} \cdot y_{ij}$, where~$y_{ij}$ follows a log-normal distribution~$ln(y_{ij})\sim \mathcal{N}(0,\sigma^2)$.
Note that the predictions do not consider slowdown effects on real hardware, e.g., due to
shared resource contention,
adding additional inaccuracy.

Additionally, we perform synthetic experiments (\Cref{sec:experiment_appendix}), which use similar workload and core configurations, but are only simulated. An advantage is that rates must not be transformed to actual schedules. The results are in line with the results of our hardware experiments.

\paragraph{Algorithms}
We consider all algorithms presented in previous sections. Additionally, we consider Round Robin (RR), which distributes a job evenly over all machines, and Iterative Greedy (\Cref{alg:greedy-abs}), which at any time iteratively schedules the job~$j$ on machine~$i$ which has the maximum~$\hs_{ij}$ among all unassigned alive jobs and free machines. We show that Iterative Greedy is not competitive (lower bound of~$\Omega(n)$).

\begin{algorithm}[tb]
  \caption{Iterative Greedy}
  \label{alg:greedy-abs}
  \begin{algorithmic}[1]
      \REQUIRE time $t$, speed predictions $\{\hs_{ij}\}$
      \STATE $I' \leftarrow I, J' \leftarrow J(t)$
      \WHILE{$I' \neq \emptyset \land J' \neq \emptyset$}
          \STATE $(i,j) = \argmax_{i \in I', j \in J'} w_j \hs_{ij}$
          \STATE $I' \leftarrow I' \setminus \{i\}, J' \leftarrow J' \setminus \{j\}$
          \STATE Schedule job $j$ on machine $i$ with rate $y_{ijt} = 1$ at time $t$.
      \ENDWHILE
  \end{algorithmic}
\end{algorithm}

\begin{lemma}\label{obs:greedy-abs-bad}
  \Cref{alg:greedy-abs} has a competitive ratio of at least $\Omega(n)$ for minimizing the total completion time on unrelated machines, even if $s_{ij} = \hs_{ij}$ for all jobs $j$ and machines $i$.
\end{lemma}

\begin{proof}
Let~$\epsilon > 0~$ and~$n > m \geq 2$ such that~$\frac{n-1}{m-1}$ is an integer. Consider a unit-weight instance of one job with~$p_1 = \frac{n-1}{m-1}$,~$s_{11} = 1 + \epsilon$ and~$s_{i1} = 1$ for~$2 \leq i \leq m$, and~$n-1$ jobs with~$p_j = \epsilon$ and~$s_{1j} = 1$ and~$s_{ij} = \epsilon$ for~$2 \leq j \leq n, 2 \leq i \leq m$. \Cref{alg:greedy-abs} first schedules job 1 on machine 1, and the~$n-1$ others on the remaining~$m-1$ machines. Since the completion time of job~$1$ is equal to~$\frac{n-1}{(1+\epsilon)(m-1)}$, jobs~$2,\ldots,n$ will complete at time at least~$\frac{n-1}{m-1}$ only on machines~$2,\ldots,m$ if~$\epsilon < \frac{m}{n-m-1}$, hence this allocation will remain until the end of the instance. This implies a total completion time of~$\Omega(\frac{n^2}{m})$ for jobs~$2,\ldots,n$. Another solution is to schedule all jobs~$2,\ldots,n$ on machine 1 with a total completion time of at most~$\bigO(\epsilon n^2)$, and job~$1$ latest at time~$\bigO(\frac{n}{m})$ on any other machine. This implies that \Cref{alg:greedy-abs} has a competitive ratio of at least~$\Omega(n)$.
\end{proof}

\paragraph{Results}

\cref{fig:results_hikey} presents the results of the hardware experiments.
We exclude PF because it produces fractional schedules which are often difficult to convert into real schedules~\cite{ImKM18}, and Greedy WSPT, because, given incorrect predictions, it significantly underperforms in synthetic experiments. We repeat each experiment 3 times with the same workload (jobs and arrival times) but different random noisy speed predictions and plot the average and standard deviation of the average completion times.

Under low system load (\cref{fig:results_hikey}a), the number of active jobs is mostly $\le$~4, i.e., it is mostly feasible to only use the \bigc cores. Consequently, the algorithms that exploit the speed-ordered property (red) consistently perform best.
Algorithms with speed predictions (blue) perform equally well for accurate predictions but their performance deteriorates for very noisy predictions. RR always uses all cores and thus shows a low performance.

Under high system load (\cref{fig:results_hikey}b), the number of active jobs is mostly $>$~4, thus, \littlec cores have to be used. RR and speed-ordered RR perform similarly, as both mostly use the same cores.
For low prediction noise ($\sigma<1$), Maximum Density performs best, but also requires most information (speed predictions and clairvoyant).
For higher prediction noise, speed-ordered Maximum Density is better 
because too noisy speed predictions result in bad schedules.
Iterative Greedy performs best among the non-clairvoyant algorithms, but does not offer any theoretical guarantees.

\begin{figure}
	\centering
    \includegraphics{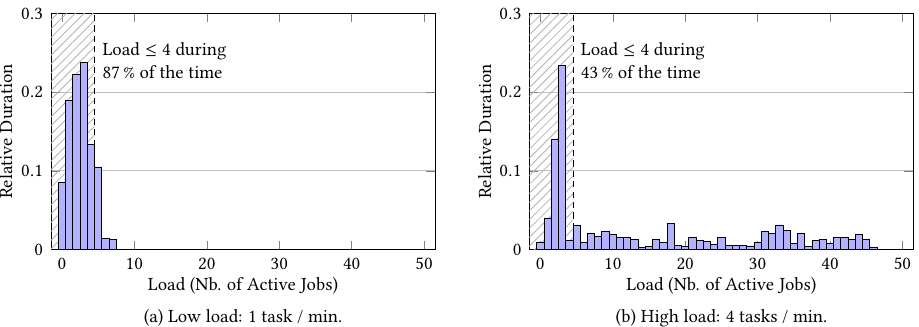}
	\caption{Distribution of the system load with speed-ordered Round Robin (\Cref{alg:rr-speed-order}).}
	\label{fig:load_durations}
\end{figure}

\paragraph{Load analysis}
\cref{fig:load_durations} shows the distribution of system load during the experiments with speed-ordered Round Robin (\Cref{alg:rr-speed-order}).
At low job arrival rate (1 task/min), the system load is $\le$~4 during 87\,\% of the time.
This means that during the majority of the time, it is possible to only use the \bigc cores, explaining why speed predictions or clairvoyance bring little benefit over the speed-ordered setting as in \cref{fig:results_hikey}a.
In contrast, the system load is $\le$~4 during 43\,\% of the time at a high job arrival rate (4 tasks/min), reaching up to 46.
Accurate speed and job volume predictions are much more beneficial in this case, explaining the larger differences between algorithms in \cref{fig:results_hikey}b.

\paragraph{Summary}
Speed predictions are beneficial in the average case if they are relatively accurate.
With inaccurate predictions, relying on the speed-ordering instead is beneficial.
\emph{In summary, our experiments show the power of speed predictions and speed-ordering for online scheduling in real-world settings.} 

\section{Conclusion and Future Directions}

We initiated research on speed-oblivious algorithms with two models motivated by real-world observations.
Future directions include settling the asymptotic competitive ratio for (non-)clairvoyant speed-oblivious algorithms on speed-ordered unrelated machines, shrinking the upper bound of PF to a small constant, and investigating speed-oblivious algorithms for other objective functions such as the total flow time, potentially also in the speed-scaling model.

\bibliographystyle{alphaurl}
\bibliography{../literature}

\newpage

\appendix

\section{Details on Algorithms with Speed Predictions}

\subsection{Full Analysis of Greedy WSPT with Speed Predictions} \label{app:min-increase}

In this section, we present an error-dependent competitive ratio for Greedy~WSPT with speed predictions and eventually prove~\Cref{theorem:minincrease}. The analysis is inspired by~\cite{GuptaMUX20}, but uses a different approach for proving the feasibility of the crafted duals. In particular, we need less scaling parameters than Gupta et al. 
\thmGreedyWSPT*

Fix an instance and the algorithm's schedule. Let $\kappa \geq 1$ and $0 < \theta < 1$ be constants.
We assume w.l.o.g.\ by scaling the instance that all processing requirements and release dates are integer multiples of $\kappa$.
Recall that $\hdense_{ij} = \frac{w_j \hs_{ij}}{p_{j}}$ and~$\hr_{ij} = \max\{r_j, \theta \frac{p_j}{\hs_{ij}}\}$.
We write for every job $j$ and machine~$i$
\begin{equation*}
	Q_{ij} = w_{j} \Bigg( \hr_{ij} + \mu_1 \frac{\hr_{ij}}{\theta} + \frac{p_j}{s_{ij}} + \sum_{\substack{j' \in M_i(j) \\ \hdense_{ij'} \geq \hdense_{ij}}} \frac{p_{j'}}{s_{ij'}} \Bigg) + \frac{p_j}{s_{ij}} \sum_{\substack{j' \in M_i(j) \\ \hdense_{ij'} < \hdense_{ij}}} w_{j'}.
\end{equation*}
Also, recall that the algorithm uses the values $\hQ_{ij}$ to assign a job $j$ at time $r_j$ to machine $g(j) = \argmin_i  \hQ_{ij} $:
\[
	\hQ_{ij} = w_{j} \Bigg( \hr_{ij} + \frac{\hr_{ij}}{\theta} + \frac{p_j}{\hs_{ij}} + \sum_{\substack{j' \in M_i(j) \\ \hdense_{ij'} \geq \hdense_{ij}}} \frac{p_{j'}}{\hs_{ij'}} \Bigg) + \frac{p_j}{\hs_{ij}} \sum_{\substack{j' \in M_i(j) \\ \hdense_{ij'} < \hdense_{ij}}} w_{j'}.
\]

We now introduce a linear programming relaxation of our problem.
As we consider a non-preemptive scheduling problem here, we can define a stronger linear program relaxation than~\eqref{pf-lp}~\cite{SchulzS02}:
  \begin{alignat}{3}
    \text{min} \quad &\sum_{i,j,t} w_{j} \cdot x_{ijt} \cdot \left( \frac{1}{2} + \frac{s_{ij}}{p_{j}} \cdot \left( t + \frac{1}{2} \right) \right) \tag{$\text{NP-LP}$}\label{mi-lp} \\
    \text{s.t.} \quad & \sum_{i,t \geq r_j} \frac{x_{ijt}s_{ij}}{p_{j}} \geq 1   &&\forall j \notag \\
    & \sum_{j} x_{ijt} \leq 1   &&\forall i, t \notag \\
    & x_{ijt} \geq 0 &&\forall i,j,t \notag \\
    & x_{ijt} = 0 &&\forall i,j,t < r_j \notag
\end{alignat}

This relaxation has an integrality gap of~$2$~\cite{SchulzS02}.
The dual of~\eqref{mi-lp} can be written as follows:
\begin{alignat}{3}
    \text{max} \quad & \quad \sum_{j} \dualVa_j - \sum_{i,t} \dualVb_{it} \tag{$\text{NP-DLP}$} \label{mi-dual} \\
    \text{s.t.} \quad & \frac{\dualVa_j s_{ij} }{p_{j}} - \dualVb_{it} \leq w_{j} \left(s_{ij} \frac{t + 1/2}{p_{j}} + \frac{1}{2} \right) \quad \forall i,j,t \geq r_j \label{constr:mi-dual} \\
    &\dualVa_j, \dualVb_{it}  \geq 0 \qquad \forall i,j,t \notag
\end{alignat}

We define a solution for~\eqref{mi-dual} which depends on the schedule produced by the algorithm. Let~$U_{i}(t) = \{ j \in J \mid g(j) = i \land t < C_j \}$.
Note that~$U_{i}(t)$ includes unreleased jobs at time $t$. Consider the following dual assignment: \begin{itemize}
  \item $\dualSa_j = Q_{g(j)j}$ for every job~$j$ and
  \item $\dualSb_{it} = \mu \cdot \sum_{j \in U_{i}(\kappa \cdot t)} w_{j}$ for every machine~$i$ and time~$t$.
\end{itemize}

We first show that the objective value of~\eqref{mi-dual} for~$(\dualSa_j,\dualSb_{it})$ is close to the objective value of the algorithm.

\begin{lemma}\label{lemma:mi-dual-objective1}
$\sum_j \dualSa_j \geq \alg$
\end{lemma}

\begin{proof}
Consider the algorithm's schedule. Let $x_i(t)$ denote the amount of time (not volume) the currently processed job on machine $i$ requires to complete. If there is no job running on machine $i$ at time $t$, we define $x_i(t) = 0$. We now calculate the contribution of some job $j$ to the algorithm's objective value $\alg$. Suppose that $j$ gets assigned to $g(j) = i$. Then, $j$ might delay other jobs with smaller predicted density which have been already assigned to $i$, i.e., are part of $M_i(j)$. Further, $j$ might be delayed by jobs which have higher predicted density and are part of $M_i(j)$. Finally, $j$'s completion time cannot be less than $\hr_{ij} + \frac{p_j}{s_{ij}}$ due to the definition of the algorithm, and this value might be delayed further by $x_i(\hr_{ij})$. In total, we conclude that the contribution of $j$ to $\alg$ is at most
\[
  w_{j} \Bigg( \hr_{ij} + x_i(\hr_{ij}) + \frac{p_j}{s_{ij}} + \sum_{\substack{j' \in M_i(j) \\ \hdense_{ij'} \geq \hdense_{ij}}} \frac{p_{j'}}{s_{ij'}} \Bigg) + \frac{p_j}{s_{ij}} \sum_{\substack{j' \in M_i(j) \\ \hdense_{ij'} < \hdense_{ij}}} w_{j'}.
  \]
This value is indeed at most $Q_{ij}$, because if at time $\hr_{ij}$ some job $k$ is being processed, it must be that $\hr_{ik} \leq \hr_{ij}$, and thus
\[
  x_i(\hr_{ij}) \leq \frac{p_k}{s_{ik}} \leq \mu_1 \frac{p_k}{\hs_{ik}} \leq \mu_1 \frac{\hr_{ik}}{\theta} \leq \mu_1 \frac{\hr_{ij}}{\theta}.
\]
The statement then follows by summation of all jobs and the observation that this contribution only affects jobs that were handled before job $j$.
\end{proof}

\begin{lemma}\label{lemma:mi-dual-objective2}
  $\sum_{i,t} \dualSb_{it} = \frac{\mu}{\kappa} \alg$
\end{lemma}

\begin{proof}
  Since we assumed that all release dates and processing times in $J$ are integer multiples of $\kappa$, all all job completions occur at integer multiples of $\kappa$. Thus, $\sum_{t} \sum_{j \in U_i(\kappa \cdot t)} w_j = \frac{1}{\kappa} \sum_{t} \sum_{j \in U_i(t)} w_j$ for every machine~$i$, and we conclude
  \[
    \sum_{i,t} \dualSb_{it} = \mu \sum_{i,t} \sum_{j \in U_{i}(\kappa \cdot t)} w_{j} =  \frac{1}{\kappa} \sum_{i,t} \sum_{j \in U_{i}(t)} w_{j} = \frac{\mu}{\kappa} \cdot \alg.
    \]
  \end{proof}

These two lemmas give the following corollary.

\begin{corollary}\label{lemma:mi-dual-objective}
    $\sum_{j} \dualSa_j - \sum_{i,t} \dualSb_{it} \geq \left( 1 - \frac{\mu}{\kappa} \right) \cdot \alg$.
 \end{corollary}

  Second, we show that scaling the crafted duals makes them feasible for~\eqref{mi-dual}. 
  
  \begin{restatable}{lemma}{lemmaMinIncreaseFeasible}\label{lemma:mi-dual-feasible}
    Assigning $\dualVa_j = \dualSa_j/\lambda$ and $\dualVb_{it} = \dualSb_{it}/\lambda$ gives a feasible solution for~\eqref{mi-dual} for a constant $\lambda > 0$ that satisfies $\lambda \geq 2 \mu(2+\theta)$ and $\lambda \geq \mu_1 (\frac{1}{\theta} + \mu_2 \cdot \kappa)$.
  \end{restatable}

  \begin{proof}
  Since our defined variables are non-negative by definition, it suffices to show that this assignment satisfies~\eqref{constr:mi-dual}.
  Fix a job~$j$, a machine~$i$ and a time~$t \geq r_j$.
  We assume that no new job arrives after~$j$, since such a job may only increase~$\dualSb_{it}$ while~$\dualSa_j$ stays unchanged.
  We define a partition of $M_i(j)$ into high priority and low priority jobs with respect to $j$, and into completed and unfinished jobs with respect to time $\kappa \cdot t$:
  \begin{itemize}
      \item $H_U = \{j' \in M_i(j): \hdense_{ij'} \geq \hdense_{ij} \land C_{j'} > \kappa \cdot t \}$ and
      $H_C = \{j' \in M_i(j): \hdense_{ij'} \geq \hdense_{ij} \land C_{j'} \leq \kappa \cdot t \}$,
      \item $L_U = \{j' \in M_i(j): \hdense_{ij'} < \hdense_{ij} \land C_{j'} > \kappa \cdot t \} $ and
      $L_C = \{j' \in M_i(j): \hdense_{ij'} < \hdense_{ij} \land C_{j'} \leq \kappa \cdot t \} $.
    \end{itemize}
  We write $H = H_C \cup H_U$, $L = L_C \cup L_U$ and $\dense_{ij} = \frac{w_j s_{ij}}{p_{j}}$.
  Due to the choice of $g(j)$ in the algorithm,~$\hQ_{g(j)j} \leq \hQ_{i'j} $ for every machine $i'$.
  Hence, we have~$\dualSa_j = Q_{g(j)j} \leq \mu_1 \cdot \hQ_{g(j)j} \leq \mu_1 \cdot \hQ_{ij}$, and using that,
  \begin{align*}
\frac{\dualSa_j \cdot s_{ij}}{\lambda p_{j}}
   &\leq \mu_1 \frac{\hQ_{ij} \cdot s_{ij}}{\lambda p_{j}} \\
   &= \dense_{ij} \frac{\mu_1}{\lambda} \left(\hr_{ij} + \frac{\hr_{ij}}{\theta} + \frac{p_j}{\hs_{ij}} + \sum_{j' \in H} \frac{p_{j'}}{\hs_{ij'}} \right)
   + \frac{\mu_1}{\lambda}  \frac{s_{ij}}{\hs_{ij}} \sum_{j' \in L} w_{j'} \\
   &\leq \dense_{ij} \frac{\mu_1}{\lambda} \left(\left(1 + \frac{1}{\theta} \right) r_j + \sum_{j' \in H} \frac{p_{j'}}{\hs_{ij'}} \right)
   + \mu_1 \frac{s_{ij} w_j}{\lambda p_j} (2 + \theta) \frac{p_j}{\hs_{ij}}
   + \frac{\mu_1}{\lambda}  \frac{s_{ij}}{\hs_{ij}} \sum_{j' \in L} w_{j'} \\
   &\leq \dense_{ij} \frac{\mu_1}{\lambda} \left(\left(1 + \frac{1}{\theta} \right) r_j + \sum_{j' \in H} \frac{p_{j'}}{\hs_{ij'}} \right)
   + \mu \frac{w_j}{\lambda} \left(2 + \theta \right)
   + \frac{\mu_1}{\lambda}  \frac{s_{ij}}{\hs_{ij}} \sum_{j' \in L} w_{j'} \\
   &\leq \dense_{ij} \frac{\mu_1}{\lambda} \left(\left(1 + \frac{1}{\theta} \right) r_j + \sum_{j' \in H} \frac{p_{j'}}{\hs_{ij'}} \right) + \frac{w_j}{2} + \frac{\mu_1}{\lambda} \frac{s_{ij}}{\hs_{ij}} \sum_{j' \in L} w_{j'},
  \end{align*}
  where the second inequality is due to $(1 + \frac{1}{\theta})\hr_{ij} \leq (1+\frac{1}{\theta})r_j + (1+\theta) \frac{p_j}{\hs_{ij}}$, which follows from the definition of $\hr_{ij}$, and the last inequality requires $\lambda \geq 2 \mu(2+\theta)$.
  Thus, asserting the dual constraint~\eqref{constr:mi-dual} reduces to proving
  \begin{equation*}
    \dense_{ij} \frac{\mu_1}{\lambda} \left(\left(1 + \frac{1}{\theta}\right) r_j+ \sum_{j' \in H} \frac{p_{j'}}{\hs_{ij'}} \right) + \frac{\mu_1}{\lambda} \frac{s_{ij}}{\hs_{ij}} \sum_{j' \in L} w_{j'} \leq   \dense_{ij} t + \frac{\dualSb_{it}}{\lambda}.
  \end{equation*}
  
  To this end, first note that for all $j' \in L$ holds
  \begin{equation}
    w_{j'} \frac{ \hs_{ij'}}{p_{j'}} = \hdense_{ij'} < \hdense_{ij} = \frac{w_j \hs_{ij}}{p_j} =  \frac{\dense_{ij} \hs_{ij}}{s_{ij}} \Longrightarrow \frac{s_{ij}}{\dense_{ij} \hs_{ij}} w_{j'} \leq \frac{\hs_{ij'}}{p_{j'}},
    \label{eq:mi-eq-for-L}
  \end{equation}
  and for all $j' \in H$
  \begin{equation}
    \dense_{ij} \leq \mu_2 \cdot \hdense_{ij} \leq \mu_2 \cdot \hdense_{ij'} = \frac{w_{j'} \hs_{ij'}}{p_{j'}} \Longrightarrow \dense_{ij} \frac{p_j'}{\hs_{ij'}} \leq w_{j'}. \label{eq:mi-eq-for-H}
  \end{equation}
  
  Using these two inequalities gives
  \begin{align*}
    & \dense_{ij} \frac{\mu_1}{\lambda} \left(\left(1 + \frac{1}{\theta}\right) r_j + \sum_{j' \in H_C} \frac{p_{j'}}{\hs_{ij'}} + \sum_{j' \in H_U} \frac{p_{j'}}{\hs_{ij'}} \right)
   + \frac{\mu_1}{\lambda} \frac{s_{ij}}{\hs_{ij}} \sum_{j' \in L_C} w_{j'}
   + \frac{\mu_1}{\lambda} \frac{s_{ij}}{\hs_{ij}} \sum_{j' \in L_U} w_{j'}  \\
   & = \dense_{ij} \frac{\mu_1}{\lambda} \left(\left(1 + \frac{1}{\theta}\right) r_j + \sum_{j' \in H_C} \frac{p_{j'}}{\hs_{ij'}}
   + \frac{s_{ij}}{\dense_{ij} \hs_{ij}} \sum_{j' \in L_C} w_{j'}
   \right)
   + \dense_{ij} \frac{\mu_1}{\lambda} \sum_{j' \in H_U} \frac{p_{j'}}{\hs_{ij'}}
   + \frac{\mu_1}{\lambda} \frac{s_{ij}}{\hs_{ij}} \sum_{j' \in L_U} w_{j'}  \\
   &\leq \dense_{ij} \frac{\mu_1}{\lambda} \left(\left(1 + \frac{1}{\theta}\right) r_j + \sum_{j' \in H_C} \frac{p_{j'}}{\hs_{ij'}}
   + \sum_{j' \in L_C} \frac{p_{j'}}{\hs_{ij'}}
   \right)
   + \frac{\mu_1}{\lambda} \frac{s_{ij}}{\hs_{ij}} \sum_{j' \in H_U} w_{j'}
   + \frac{\mu_1}{\lambda} \frac{s_{ij}}{\hs_{ij}} \sum_{j' \in L_U} w_{j'}  \\
   &\leq \dense_{ij} \frac{\mu_1}{\lambda} \left(\frac{r_j}{\theta} + \mu_2 \left(r_j + \sum_{j' \in M_i(j): \kappa \cdot t \geq C_{j'}} \frac{p_{j'}}{s_{ij'}} \right)
   \right)
   + \frac{\mu_1}{\lambda} \mu_2 \sum_{j' \in M_i(j): \kappa \cdot t < C_{j'}} w_{j'} \\
   &\leq \dense_{ij} \frac{\mu_1}{\lambda} \left(\frac{t}{\theta}  + \mu_2 \cdot \kappa \cdot t
   \right)
   + \frac{\mu}{\lambda}\sum_{j' \in U_i(\kappa \cdot t)} w_{j'} \\
   & \leq \dense_{ij} t + \frac{\dualSb_{it}}{\lambda}.
  \end{align*}
  In the first inequality we use \eqref{eq:mi-eq-for-H} and \eqref{eq:mi-eq-for-L}. In order to understand the third inequality, first recall that $M_i(j)$ contains all jobs that are assigned to machine~$i$ but unstarted at time~$r_j$. Thus, the total processing duration of these jobs that are completed within time~$\kappa \cdot t$ can be at most~$\kappa \cdot t - r_j$. The last inequality follows from~$\lambda \geq \mu_1 (\frac{1}{\theta} + \mu_2 \cdot \kappa)$ and the definition of $\dualSb_{it}$.
\end{proof}

\begin{proof}[Proof of~\Cref{theorem:minincrease}]
  We set $\kappa = \frac{23}{6}\mu$, $\theta=\frac{2}{3}$ and $\lambda = \frac{16}{3} \mu^2$. Then, weak duality, \Cref{lemma:mi-dual-objective} and \Cref{lemma:mi-dual-feasible} imply
   \begin{align*}
   \opt \geq \sum_{j} \dualVa_j - \sum_{i,t} \dualVb_{it}
   = \frac{1}{\lambda} \left( \sum_{j} \dualSa_j - \sum_{i,t} \dualSb_{it} \right) 
     = \left( \frac{1 - \mu/\kappa}{\lambda} \right) \cdot \alg.
   \end{align*}
   Since $\kappa > \mu$ and $\lambda > 0$, we conclude that
   \[
    \alg \leq \frac{\frac{16}{3} \cdot \mu^2}{1 - \frac{6}{23}} \cdot \opt = \frac{368}{51} \cdot \mu^2 \cdot \opt.
   \]
 \end{proof}

\subsection{Full Analysis of Proportional Fairness with Speed Predictions}\label{app:pf}

This section contains the detailed analysis of PF with speed predictions, and thus the proof of~\Cref{thm:pf}. It is based on the analysis of the speed-aware PF given in~\cite{ImKM18}. 

\thmPF*

Fix an instance and PF's schedule.
Let~$\kappa \geq 1$ and~$0 < \lambda < 1$ be constants which we fix later.
Recall that~$q_{jt}$ denotes the progress of job~$j$ at time~$t$. For every~$t$, consider the sorted (ascending) list~$Z^t$ composed of $w_j$ copies of~$\frac{q_{jt}}{p_j}$ for every~$j \in U_t$. Note that $Z^t$ has length $W_t$.
We define~$\zeta_t$ as the value at the index~$\floor{\lambda W_t}$ in~$Z^t$.

We first state the KKT conditions with multipliers $\{\eta_{it}\}_{i}$ and $\{\theta_{jt}\}_{j \in J(t)}$ of the optimal solution $\{y_{ijt}\}_{i,j}$ of \eqref{pf-convex} the algorithm uses at time $t$:

\begin{align}
  \frac{\hs_{ij} w_j}{\sum_{i'} \hs_{i'j} y_{i'jt}} &\leq  \theta_{jt} + \eta_{it} \quad \forall t, \forall i, \forall j\in J(t) \label{pf-kkt1}\\
  y_{ijt} \left( \frac{\hs_{ij} w_j}{\sum_{i'} \hs_{i'j} y_{i'jt}} - (\theta_{jt} + \eta_{it}) \right) &= 0 \quad \forall t, \forall i, \forall j\in J(t) \label{pf-kkt12}\\
    \theta_{jt} \left( \sum_i y_{ijt} - 1 \right) &= 0 \quad \forall t, \forall j \in J(t) \label{pf-kkt2} \\
    \eta_{it} \left( \sum_j y_{ijt} - 1 \right) &= 0 \quad \forall t, \forall i  \label{pf-kkt3} \\
    \theta_{jt}, \eta_{it} & \geq 0  \quad \forall t, \forall i, \forall j\in J(t)
\end{align}

We have the following dual assignment:
\begin{itemize}
  \item $\dualSa_j = \sum_{t' = 0}^{C_j} \dualSa_{jt'}$, where $\dualSa_{jt'} = w_j \cdot \ind \left[ \frac{q_{jt'}}{p_j} \leq \zeta_{t'} \right]$, for every job $j$,
  \item $\dualSb_{it} = \frac{1}{\kappa} \sum_{t' \geq t}  \zeta_{t'} \eta_{it'}$ for every machine $i$ and time $t$, and
  \item $\dualSc_{jt} = \frac{1}{\kappa} \sum_{t' = t}^{C_j}  \zeta_{t'} \theta_{jt'}$ for every job $j$ and time $t \geq r_j$.
\end{itemize}

The following three lemmas will conclude that the dual objective value of this assignment is close the algorithm's objective value, and thus prove~\Cref{lemma:pf-dual-objective}.

\begin{lemma}\label{lemma:pf-dual-obj1}
    $\sum_j \dualSa_j \geq \lambda \cdot \alg$
\end{lemma}

\begin{proof}
Consider a time~$t$ and the list~$Z^t$. Observe that~$\sum_{j \in U_t} \dualSa_{jt}$ contains for every job~$j$ which satisfies~$\frac{q_{jt}}{p_j} \leq \zeta_{t}$ its weight~$w_j$. By the definitions of~$Z_t$ and~$\zeta_t$, we conclude that this is at least~$\lambda W_t$, i.e.,~$\sum_{j \in U_t} \dualSa_{jt} \geq \lambda W_t$.
  The statement then follows by summing over all times~$t$.
\end{proof}

\begin{lemma}\label{lemma:kkt-multiplier-bound}
  At any time $t$, $\sum_i \eta_{it}  + \sum_{j \in J(t)} \theta_{jt} \leq W_t$.
\end{lemma}

\begin{proof}
  At any time $t$ holds
  \begin{align*}
    \sum_i \eta_{it} + \sum_{j \in J(t)} \theta_{jt}
    &= \left( \sum_i \eta_{it} \sum_{j\in J(t)} y_{ijt} \right) + \left( \sum_{j \in J(t)} \theta_{jt} \sum_{i} y_{ijt} \right) \\
    &= \sum_{i} \sum_{j\in J(t)}  y_{ijt} (\eta_{it} + \theta_{jt}) \\
    &= \sum_{i} \sum_{j\in J(t)}  y_{ijt} \frac{\hs_{ij} w_j}{\sum_{i'} \hs_{i'j} y_{i'jt}} \\
    &= \sum_{j\in J(t)} \sum_i \hs_{ij} y_{ijt} \frac{w_j}{\sum_{i'} \hs_{i'j} y_{i'jt}} = \sum_{j \in J(t)} w_j \leq W_t.
  \end{align*}
  The first equality is due to~\eqref{pf-kkt2} and~\eqref{pf-kkt3}, and the third equality due to~\eqref{pf-kkt12}.
\end{proof}

\begin{lemma}\label{lemma:pf-dual-obj2}
    At any time $t$, $\sum_i \dualSb_{it} + \sum_{j \in J: t \geq r_j} \dualSc_{jt} \leq \frac{4}{(1-\lambda) \kappa} W_t$.
\end{lemma}
\begin{proof}
  Fix a time $t$. Observe that for every $t' \geq t$ the definitions of $Z_{t'}$ and $\zeta_{t'}$ imply $(1 - \lambda) W_{t'} \leq \sum_{j \in U_{t'}} w_j \cdot \ind \left[ \frac{q_{jt'}}{p_j} \geq \zeta_{t'} \right]$. Thus,
  \begin{equation}  
    \zeta_{t'} \cdot (1 - \lambda) W_{t'} \leq  \sum_{j \in U_{t'}} w_j \cdot \zeta_{t'} \cdot \ind \left[ \frac{q_{jt'}}{p_j} \geq \zeta_{t'} \right] \leq \sum_{j \in U_{t'}} w_j \cdot \frac{q_{jt'}}{p_j} \cdot \ind \left[ \frac{q_{jt'}}{p_j} \geq \zeta_{t'} \right]. \label{eq:pf-dual-obj2-eq1}
  \end{equation}
  
We define a partition $\{M_k\}_{k \geq 1}$ of the time interval $[t, \infty)$ such that the total weight of unfinished jobs at all times during~$M_k$ is part of $(\frac{1}{2^k} W_t, \frac{1}{2^{k-1}} W_t]$. Fix a $k \geq 1$. Rearranging \eqref{eq:pf-dual-obj2-eq1} and estimating the total weight of unfinished jobs in a partition against both its upper and lower bound yields
  \begin{align*}
    \sum_{t' \in M_k} \zeta_{t'} &\leq \sum_{t' \in M_k} \frac{1}{1 - \lambda} \sum_{j \in U_{t'}} \frac{w_j}{W_{t'}} \cdot \frac{q_{jt'}}{p_j} \cdot \ind \left[ \frac{q_{jt'}}{p_j} \geq \zeta_{t'} \right] \\
    &\leq \frac{1}{1 - \lambda} \sum_{t' \in M_k}  \sum_{j \in U_{t'}} \frac{w_j}{W_{t'}} \cdot \frac{q_{jt'}}{p_j}  \\
    &\leq \frac{2^k}{(1 - \lambda) W_t} \sum_{t' \in M_k}  \sum_{j \in U_{t'}} w_j \cdot \frac{q_{jt'}}{p_j}  \\
    &\leq \frac{2^k \cdot W_t}{(1 - \lambda) W_t \cdot 2^{k-1}} = \frac{2}{1 - \lambda}.
  \end{align*}

  The definitions of $\dualSb_{it}$ and $\dualSc_{jt}$ and \Cref{lemma:kkt-multiplier-bound} imply
  \begin{align*}
    \sum_i \dualSb_{it} + \sum_{j \in J: t \geq r_j} \dualSc_{jt} &= \left( \sum_i \frac{1}{\kappa}\sum_{t' \geq t}   \eta_{it'} \cdot \zeta_{t'} \right)  + \left( \sum_{j \in J: t \geq r_j} \frac{1}{\kappa} \sum_{t' = t}^{C_j}  \theta_{jt'} \cdot \zeta_{t'} \right) \\
    &=  \frac{1}{\kappa} \sum_{t' \geq t} \zeta_{t'} \left(\sum_i \eta_{it'} + \sum_{j \in J(t')} \theta_{jt'} \right) \leq \frac{1}{\kappa} \sum_{t' \geq t} \zeta_{t'} W_{t'}.
  \end{align*}

  By dividing the time after $t$ into the partition $\{M_k\}_{k \geq 1}$ and using our bound on $\sum_{t' \in M_k} \zeta_{t'}$, we conclude that this is at most
  \begin{align*}
    \frac{1}{\kappa} \sum_{k \geq 1} \sum_{t' \in M_k} \zeta_{t'} W_{t'} 
    \leq \frac{1}{\kappa} \sum_{k \geq 1} \frac{W_t}{2^{k-1}} \sum_{t' \in M_k} \zeta_{t'}
    \leq \frac{2}{\kappa (1-\lambda)} W_t \sum_{k \geq 1}  \frac{1}{2^{k-1}} 
    \leq \frac{4}{\kappa (1-\lambda)} W_t.
  \end{align*}
  The last inequality uses a bound on the geometric series.
\end{proof}

\pfDualObj*

\begin{proof}
  Follows directly from~\Cref{lemma:pf-dual-obj1,lemma:pf-dual-obj2}.
\end{proof}

\pfDualFeasibility*

\begin{proof}
    First observe that for every $t$ and $j$ holds
    \begin{equation}
      \sum_i \hs_{ij} y_{ijt} \leq \mu_1 \sum_i s_{ij} y_{ijt} = \mu_1 \cdot q_{jt}. \label{eq:pred-to-real-progress}
    \end{equation}

    Fix a job $j$, a machine $i$ and a time $t \geq r_j$. 
    \begin{align*}
        \frac{ \dualSa_j s_{ij}}{p_j} - w_j \cdot \frac{t \cdot s_{ij}}{p_j}
        &\leq  s_{ij} \cdot \sum_{t' = t}^{C_j} \frac{\dualSa_{jt'}}{p_j} \\
        &=   s_{ij} \cdot \sum_{t' = t}^{C_j} \frac{w_j}{p_j} \cdot \ind\left[ \frac{q_{jt'}}{p_j} \leq \zeta_{t'} \right] \\
        &=   s_{ij} \cdot  \sum_{t' = t}^{C_j} \frac{w_j}{\sum_{i'} \hs_{i'j} y_{i'jt'}} \cdot \frac{\sum_{i'} \hs_{i'j} y_{i'jt'}}{q_{jt'}}  \cdot \frac{q_{jt'}}{p_j} \cdot \ind\left[ \frac{q_{jt'}}{p_j} \leq \zeta_{t'} \right] \\
        &\leq \mu_1 \cdot \mu_2 \cdot \sum_{t' = t}^{C_j} \frac{\hs_{ij} w_j}{\sum_{i'} \hs_{i'j} y_{i'jt'}}  \cdot \frac{q_{jt'}}{p_j} \cdot \ind\left[ \frac{q_{jt'}}{p_j} \leq \zeta_{t'} \right] \\
        &\leq   \mu \cdot \sum_{t' = t}^{C_j} \left( \eta_{it'} + \theta_{jt'} \right) \cdot \frac{q_{jt'}}{p_j} \cdot \ind\left[ \frac{q_{jt'}}{p_j} \leq \zeta_{t'} \right] \\
        &\leq   \mu \cdot \sum_{t' = t}^{C_j} \left( \eta_{it'} + \theta_{jt'} \right) \cdot \zeta_{t'} \\
        &\leq  \mu \kappa \cdot  \frac{1}{\kappa} \left(\sum_{t' \geq t}  \eta_{it'} \cdot \zeta_{t'} \right) + \mu \kappa \cdot \left(  \frac{1}{\kappa} \sum_{t' = t}^{C_j} \theta_{jt'} \cdot \zeta_{t'} \right) \\
        &=  \mu \kappa \cdot \dualSb_{it}  + \mu \kappa \cdot \dualSc_{jt}.
    \end{align*}

    The second inequality uses~\eqref{eq:pred-to-real-progress} and the third inequality uses~\eqref{pf-kkt1}. Since $\alpha = \kappa \mu$, this dual assignment indeed satisfies the constraint of \eqref{pf-dual}.
\end{proof}

\section{Details on Round Robin for Speed-Ordered Machines}\label{app:rr-speed-ordered}

\subsection{Missing Details for the Analysis for Unrelated Machines}\label{app:rr-speed-ordered-unrelated}

This section contains missing details for the proof of \Cref{thm:unrelated-speed-order}, which we firstly restate:

\thmRoundRobinSOunrelated*

\begin{proposition}\label{rr-so-prop1}
  At any time $t$, $\sum_i \beta_{it} \leq \bigO(\log(\min\{n,m\})) \cdot \abs{U_t}$.
\end{proposition}

\begin{proof}
  At any time $t$,
  \[
    \sum_{i \in I} \beta_{it} = \sum_{i = 1}^{m_{t}} \frac{1}{i} \cdot \abs{J(t)} \leq \abs{U_t} \sum_{i = 1}^{m_{t}} \frac{1}{i} \leq \bigO(\log(\min\{n,m\})) \cdot \abs{U_t},
  \]
  where in the last inequality we use that $m_t = \min\{m, \abs{J(t)}\} \leq \min\{m, n\}$.
\end{proof}

\begin{proposition}\label{rr-so-prop2}
  At any time $t$, $\sum_{j \in J: r_j \geq t} \gamma_{jt} \leq \abs{U_t}$.
\end{proposition}

\begin{lemma}\label{lemma:so-rr-dual-obj1}
  $\sum_{j} \dualSa_j \geq \frac{1}{2} \cdot \alg$.
\end{lemma}

\begin{proof}
  Analogous to the proof of~\Cref{lemma:pf-dual-obj1}.
\end{proof}

\begin{lemma}\label{lemma:so-rr-dual-obj2}
  At any time $t$, $\sum_{i} \dualSb_{it} \leq \bigO(1) \cdot \abs{U_t}$.
\end{lemma}
 
\begin{proof}
  Analogous to the proof of~\Cref{lemma:pf-dual-obj2} when using~\Cref{rr-so-prop1} and the fact that $\kappa = \Theta(\log(\min\{m, n\}))$.
\end{proof}

\begin{lemma}\label{lemma:so-rr-dual-obj3}
  At any time $t$, $\sum_{j \in J: r_j \geq t} \dualSc_{jt} \leq \bigO(1) \cdot \abs{U_t}$.
\end{lemma}

\begin{proof}
Analogous to the proof of~\Cref{lemma:pf-dual-obj2} when using~\Cref{rr-so-prop2}.
\end{proof}

Observe that \Cref{lemma:so-rr-dual-obj1}, \Cref{lemma:so-rr-dual-obj2} and \Cref{lemma:so-rr-dual-obj3} imply \Cref{lemma:rr-speed-order-dual-objective-unrelated}. It remains the proof of~\Cref{thm:unrelated-speed-order}:

\begin{proof}[Proof of~\Cref{thm:unrelated-speed-order}]
  Weak duality, \Cref{lemma:rr-speed-order-dual-objective-unrelated} and \Cref{lemma:rr-speed-order-dual-feasibility-unrelated} imply
  \begin{align*}
      \kappa \cdot \opt &\geq \opt_{\kappa}
      \geq  \sum_{j} \dualSa_j - \sum_{i,t} \dualSb_{it} - \sum_{j,t \geq r_j} \dualSc_{jt} \geq  \Omega(1) \cdot \alg.
  \end{align*}
  We conclude the proof by noting that $\kappa = \Theta(\log(\min\{m, n\}))$.
\end{proof}

\subsection{Full Analysis of Round Robin for Speed-Ordered Related Machines}\label{app:rr-speed-ordered-related}

\begin{restatable}{theorem}{thmRoundRobinSO}\label{thm:related-speed-order}
  \Cref{alg:rr-speed-order} has a competitive ratio of at most~216 for minimizing the total completion time on speed-ordered related machines.
\end{restatable}
  
We prove this theorem using a dual-fitting proof based on~\eqref{pf-dual}, where~$w_j = 1$ and~$s_i = s_{ij}$ for every job~$j$ and every machine~$i$. 
Fix an instance and the algorithm's schedule. For every time~$t$ we write~$m_t = \min\{m, \abs{J(t)}\}$.
We define for every machine~$i$ and any time~$t$
\[
    \beta_{it} = \frac{s_i}{\sum_{\ell=1}^{m_t} s_\ell} \cdot \abs{J(t)} \cdot \ind\left[ i \leq \abs{J(t)}  \right],
\]
  and \(
    \gamma_{jt} = \ind\left[ j \in J(t) \right]
  \) for every job~$j$ and any time~$t$.

Observe the following bounds when summing up these values:

  \begin{proposition}\label{rr-so-prop1-related}
    At any time $t$, $\sum_i \beta_{it} \leq \abs{U_t}$.
  \end{proposition}
  
  \begin{proposition}\label{rr-so-prop2-related}
    At any time $t$, $\sum_{j \in J(t)} \gamma_{jt} \leq \abs{U_t}$.
  \end{proposition}

  Let~$\kappa \geq 1$ and~$0 < \lambda < 1$ be constants. For every~$t$, consider the sorted (ascending) list~$Z^t$ composed of values~$\frac{q_{jt}}{p_j}$ for every~$j \in U_t$. 
  We define~$\zeta_t$ as the value at the index~$\floor{\lambda \abs{U_t}}$ in~$Z^t$.
  Consider the following duals:
  \begin{itemize}
    \item $\dualSa_j = \sum_{t' = 0}^{C_j} \ind \left[ \frac{q_{jt'}}{p_j} \leq \zeta_{t'} \right]$ for every job $j$,
    \item $\dualSb_{it} = \frac{1}{\kappa} \sum_{t' \geq t} \beta_{it'} \zeta_{t'}$ for every $i$ and $t$, and
    \item $\dualSc_{jt} = \frac{1}{\kappa} \sum_{t' \geq t} \gamma_{jt'} \zeta_{t'}$ for every $j$ and $t \geq r_j$.
  \end{itemize}
  
  \begin{lemma}\label{lemma:so-rr-dual-obj1-related}
    $\sum_{j} \dualSa_j \geq \lambda \cdot \alg$.
  \end{lemma}
  
  \begin{proof}
    Analogous to the proof of~\Cref{lemma:pf-dual-obj1}.
  \end{proof}

  \begin{lemma}\label{lemma:so-rr-dual-obj2-related}
    At any time $t$, $\sum_{i} \dualSb_{it} \leq \frac{4}{(1-\lambda)\kappa} \abs{U_t}$.
  \end{lemma}
   
  \begin{proof}
    Analogous to the proof of~\Cref{lemma:pf-dual-obj2} when using~\Cref{rr-so-prop1-related}.
  \end{proof}

  \begin{lemma}\label{lemma:so-rr-dual-obj3-related}
    At any time $t$, $\sum_{j \in J: r_j \geq t} \dualSc_{jt} \leq \frac{4}{(1-\lambda)\kappa} \abs{U_t}$.
  \end{lemma}
  
  \begin{proof}
  Analogous to the proof of~\Cref{lemma:pf-dual-obj2} when using~\Cref{rr-so-prop2-related}.
  \end{proof}

\Cref{lemma:so-rr-dual-obj1-related,lemma:so-rr-dual-obj2-related,lemma:so-rr-dual-obj3-related} conclude the following bound between $\alg$ and the objective value of the crafted duals.

  \begin{lemma}\label{lemma:rr-speed-order-dual-objective-related}
    $(\lambda - \frac{8}{(1-\lambda)\kappa})  \alg \leq \sum_{j} \dualSa_j - \sum_{i,t} \dualSb_{it} - \sum_{j,t \geq r_j} \dualSc_{jt}$
  \end{lemma}

We finally prove that the crafted duals are feasible under certain conditions.  

  \begin{lemma}\label{lemma:rr-speed-order-dual-feasibility-related}
    Assigning $\dualVa_j = \dualSa_j$, $\dualVb_{it} = \dualSb_{it}$ and $\dualVc_{jt} = \dualSc_{jt}$ is feasible for~\eqref{pf-dual} if $\alpha = \kappa$ and $s_i = s_{ij}$ for all machines $i$ and jobs $j$.
  \end{lemma}
  
  \begin{proof}
    First observe that the dual assignment is non-negative. Let~$i \in I, j \in J$ and~$t \geq r_j$. Since the rates of \Cref{alg:rr-speed-order} imply \( q_{jt} = \sum_{\ell = 1}^{m_t} \frac{s_\ell}{\abs{J(t)}} \), we have
    \begin{align}
      \frac{\dualSa_j s_i}{p_j} - \frac{ s_i \cdot t}{p_j} \leq \sum_{t' = t}^{C_j} \frac{ s_i}{p_j} \cdot \ind\left[ \frac{q_{jt'}}{p_j} \leq \zeta_{t'} \right] = \sum_{t' = t}^{C_j} \frac{ s_i}{q_{jt'}} \cdot \frac{q_{jt'}}{p_j} \cdot \ind\left[ \frac{q_{jt'}}{p_j} \leq \zeta_{t'} \right] \leq \sum_{t' = t}^{C_j} \frac{ s_i}{\sum_{\ell = 1}^{m_{t'}} \frac{s_\ell }{\abs{J(t')}}} \cdot \zeta_{t'} \label{eq:speed-ordered-rr}
\end{align}
  
    Consider any time~$t'$ with~$t \leq t' \leq C_j$. If~$i \leq \abs{J(t')}$, the definition of~$\beta_{it'}$ yields
    \[
      \frac{s_i}{\sum_{\ell = 1}^{m_{t'}} s_\ell} \cdot \abs{J(t')} \cdot \zeta_{t'} = \beta_{it'} \cdot \zeta_{t'}.
    \]
  
    Otherwise,~$i > \abs{J(t')}$, the fact that~$s_1 \geq \ldots \geq s_m$ implies~$\sum_{\ell = 1}^{m_{t'}} s_\ell \geq \sum_{\ell = 1}^{\abs{J(t')}} s_\ell \geq \abs{J(t')} \cdot s_i$, and thus 
    \[
      \frac{s_i}{\sum_{\ell = 1}^{m_{t'}} s_\ell} \cdot \abs{J(t')} \cdot \zeta_{t'} \leq \frac{\abs{J(t')}}{\abs{J(t')}} \cdot \zeta_{t'} =  \gamma_{jt'} \cdot \zeta_{t'},
    \]
    because~$t' \leq C_j$. 
    Put together, \eqref{eq:speed-ordered-rr} is at most
    \begin{align*}
\sum_{t' = t}^{C_j} \beta_{it'}\zeta_{t'} + \sum_{t' = t}^{C_j} \gamma_{jt'}\zeta_{t'} \leq \kappa (\dualSb_{it} + \dualSc_{jt}),
  \end{align*}
  which verifies the dual constraint.
  \end{proof}
  
  \begin{proof}[Proof of~\Cref{thm:related-speed-order}]
    Weak duality, \Cref{lemma:rr-speed-order-dual-objective-related} and \Cref{lemma:rr-speed-order-dual-feasibility-related} imply
    \begin{align*}
        \kappa \cdot \opt \geq \opt_{\kappa}
        \geq  \sum_{j} \dualSa_j - \sum_{i,t} \dualSb_{it} - \sum_{j,t \geq r_j} \dualSc_{jt} \geq \left(\lambda -  \frac{8}{(1-\lambda)\kappa}\right) \cdot \alg.
    \end{align*}
    Setting $\kappa = 72$ and $\lambda = \frac{2}{3}$ concludes $\alg \leq 216 \cdot \opt$.
\end{proof}

\section{Further Details on Experimental Results}
\label{sec:experiment_appendix}

\subsection{Implementation Details}

We implemented the schedulers as separate applications running in \emph{userspace}, scheduling jobs via Linux \emph{affinity masks}, which indicate for each process on which core it may be executed.
The schedulers compute a schedule every 2\,s based on the currently active jobs and their (predicted) characteristics.
The schedulers are provided with the process IDs of the tasks in the workload, potentially along with predictions, and only manage these processes via affinity masks.
Other processes may run on any core, but their load is negligible.

We use \emph{native} input set for the \parsec jobs.
For the \splash jobs, we use both the \emph{large} input set and custom input sizes to study the impact of different input data (see~\cref{fig:characterization}).
The \polybench jobs use their standard hard-coded inputs.
We discard jobs that execute for less than 30\,s on a \bigc core to reduce measurement noise, and discard jobs that use more than 512\,MB RAM because the \hikey board has only 6\,GB RAM and Android does not support swap.
Overall, this results in 43 combinations of jobs and input data, i.e., some jobs are repeated in the workloads.

\subsection{Simulation Experiments}

The experiments on real hardware are slow, hence we can only study a limited set of workload scenarios.
We therefore additionally test many different scenarios in synthetic experiments in simulation.
These synthetic experiments also model an 8-core processor.
We create 20~random workloads with 100~synthetic jobs, whose arrival times are drawn from a Poisson distribution and with random characteristics: 4 \littlec cores with speed 1, 4 \bigc cores with job-dependent speeds from~$\mathcal{U}(2,6)$, and~$p_j \sim \mathcal{U}(60,600)$.
Speed predictions are same as in the hardware experiments, i.e.,~$\hs_{ij}=s_{ij} \cdot y_{ij}$.

\begin{figure}
	\centering
    \includegraphics{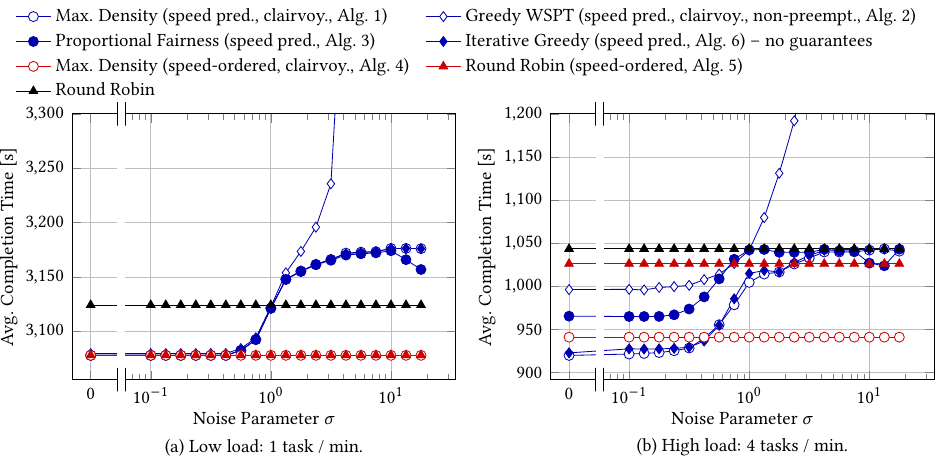}
	\caption{Synthetic experiments. The experiments are each repeated 20 times with different random workloads.}
	\label{fig:results_synth}
\end{figure}

\cref{fig:results_synth} shows the results of the synthetic experiments, including the fractional schedulers Greedy WSPT and PF.
Unlike the real experiments, we are not restricted to a single workload and instead run 20 different workloads and plot the average results.
Inaccurate speed predictions in Greedy WSPT result in large idle times, greatly deteriorating the performance.
PF performs similar to or worse than Maximum Density, depending on the system load.
The other algorithms perform similar to the experiments on the real platform, confirming that the results are not depending on a specific workload.

\end{document}